\documentclass{amsart}
\usepackage{amsaddr}
\usepackage[left=2cm,right=2cm,top=2cm,bottom=2cm]{geometry}
\usepackage{setspace}
\usepackage{amsmath}
\usepackage{amssymb}
\usepackage{bbm}

\usepackage{graphicx}
\usepackage{verbatim}
\usepackage{array}
\usepackage{booktabs}
\usepackage{threeparttable}
\usepackage[update,prepend]{epstopdf}
\usepackage{multirow}
\usepackage{amsmath}
\usepackage{amsthm}
\usepackage[usenames,dvipsnames]{xcolor}
\usepackage[hidelinks]{hyperref}
\hypersetup{
	bookmarks=true,         
	unicode=false,          
	pdftoolbar=true,        
	pdfmenubar=true,        
	pdffitwindow=false,     
	pdfstartview={FitH},    
	pdftitle={My title},    
	pdfauthor={Author},     
	pdfsubject={Subject},   
	pdfcreator={Creator},   
	pdfproducer={Producer}, 
	pdfkeywords={keyword1} {key2} {key3}, 
	pdfnewwindow=true,      
	colorlinks=true,       
	linkcolor=Red,          
	citecolor=Green,        
	filecolor=Magenta,      
	urlcolor=Blue           
}
\usepackage{doi}
\usepackage{url}
\usepackage{neuralnetwork}
\usepackage{tikz}
\usepackage{tikz-qtree}
\usetikzlibrary{shapes.multipart,matrix,shapes,arrows,positioning}
\usepackage[ backend=bibtex,
hyperref=true,
url=false,
isbn=false,
backref=false,
style=numeric-comp,
sorting=none,
block=none]{biblatex}
\bibliography{reference.bib}

\DeclareFieldFormat{journaltitle}{\mkbibemph{#1},} 
\DeclareFieldFormat[inbook,thesis]{title}{\mkbibemph{#1}\addperiod} 
\DeclareFieldFormat[article]{title}{#1} 
\DeclareFieldFormat{pages}{#1} 

\renewbibmacro*{volume+number+eid}{%
	\printfield{volume}%
	\setunit*{\addnbspace}
	\printfield{number}%
	\setunit{\addcomma\space}%
	\printfield{eid}}

\DeclareFieldFormat[article]{volume}{\textbf{#1}} 
\DeclareFieldFormat[article]{number}{\mkbibparens{#1}\addcolon} 

\renewbibmacro*{issue+date}{%
	\setunit*{\addcomma\space}
	\iffieldundef{issue}
	{\usebibmacro{date}}
	{\printfield{issue}%
		\setunit*{\addspace}%
		\usebibmacro{date}}
	\newunit}

\renewbibmacro*{journal+issuetitle}{%
	\usebibmacro{journal}%
	\setunit*{\addspace}%
	\iffieldundef{series}
	{}
	{\newunit
		\printfield{series}%
		\setunit{\addspace}}%
	\usebibmacro{volume+number+eid}%
	\newunit}

\DeclareBibliographyDriver{article}{%
	\usebibmacro{bibindex}%
	\usebibmacro{begentry}%
	\usebibmacro{author/translator+others}%
	\setunit{\labelnamepunct}\newblock
	\usebibmacro{title}%
	\newunit
	\printlist{language}%
	\newunit\newblock
	\usebibmacro{byauthor}%
	\newunit\newblock
	\usebibmacro{bytranslator+others}%
	\newunit\newblock
	\printfield{version}%
	\newunit\newblock
	\usebibmacro{journal+issuetitle}%
	\newunit
	\usebibmacro{byeditor+others}%
	\newunit
	\usebibmacro{note+pages}%
	\setunit{\addspace}
	\usebibmacro{issue+date}
	\setunit{\addcolon\space}
	\usebibmacro{issue}
	\newunit\newblock
	\iftoggle{bbx:isbn}
	{\printfield{issn}}
	{}%
	\newunit\newblock
	\usebibmacro{doi+eprint+url}%
	\newunit\newblock
	\usebibmacro{addendum+pubstate}%
	\setunit{\bibpagerefpunct}\newblock
	\usebibmacro{pageref}%
	\usebibmacro{finentry}}

\usepackage{caption}
\usepackage{subcaption}

\theoremstyle{plain}
\newtheorem{theo}{Theorem}[section]
\newtheorem{prop}{Proposition}[section]
\newtheorem{lemm}{Lemma}[section]
\newtheorem{coro}{Corrollary}[section]

\theoremstyle{definition}
\newtheorem{defn}{Definition}[section]
\newtheorem{exam}{Example}[section]

\theoremstyle{remark}
\newtheorem*{remark}{Remark}
\newtheorem*{notation}{Notation}

\numberwithin{equation}{section}

\DeclareMathOperator{\Tr}{Tr}

\begin{document}

\let\origmaketitle\maketitle
\def\maketitle{
  \begingroup
  \def\uppercasenonmath##1{} 
  \let\MakeUppercase\relax 
  \origmaketitle
  \endgroup
}

\title{\bfseries \LARGE {\scshape The Learnability of Unknown Quantum Measurements}}

\author{\large \textsc{Hao-Chung Cheng$^{1}$, Min-Hsiu Hsieh$^2$, and Ping-Cheng Yeh$^3$}}
\address{Graduate Institute Communication Engineering, National Taiwan University, Taiwan (R.O.C)$^{1,3}$ \\ Centre for Quantum Computation and Intelligent Systems, \\Faculty of Engineering and Information Technology, University of Technology Sydney, Australia$^{1,2}$} 
\email{\href{mailto:f99942118@ntu.edu.tw}{f99942118@ntu.edu.tw}$^1$}
\email{\href{mailto:minhsiuh@gmail.com}{minhsiuh@gmail.com}$^2$}
\email{\href{mailto:pcyeh@ntu.edu.tw}{pcyeh@ntu.edu.tw}$^3$}

\begin{abstract}
Quantum machine learning has received significant attention in recent years, and promising progress has been made in the development of quantum algorithms to speed up traditional machine learning tasks.
In this work, however, we focus on investigating the information-theoretic upper bounds of sample complexity---how many training  samples are sufficient to predict the future behaviour of an unknown target function. This kind of problem is, arguably, one of the most fundamental problems in statistical learning theory and the bounds for practical settings can be completely characterised by a simple measure of complexity.

Our main result in the paper is that, for learning an unknown quantum measurement, the upper bound, given by the fat-shattering dimension, is linearly proportional to the dimension of the underlying Hilbert space. Learning an unknown quantum state becomes a dual problem to ours, and as a byproduct, we can recover Aaronson's famous result [\emph{Proc.~R.~Soc.~A \textbf{463}, 3089--3144 (2007)}] solely using a classical machine learning technique. In addition,  other famous complexity measures like covering numbers and Rademacher complexities are derived explicitly. We are able to connect measures of sample complexity with various  areas in quantum information science, e.g.~quantum state/measurement tomography, quantum state discrimination and quantum random access codes, which may be of independent interest. Lastly, with the assistance of general Bloch-sphere representation,
we show that learning quantum measurements/states can be mathematically formulated as a neural network. Consequently, classical ML algorithms can be applied to efficiently accomplish the two quantum learning tasks.

\end{abstract}

\maketitle



\section{Introduction} \label{intro}

\emph{Statistical learning theory} \cite{Vap95,Vap98} or \emph{Machine Learning} (ML) \cite{Mit97} is a branch of artificial intelligence which aims to devise algorithms for machines to systematically learn from historic data. Typically, ML has been separated into \emph{unsupervised learning} and \emph{supervised learning}. In unsupervised learning, the machine is most useful for finding the hidden structure, e.g.~{clustering} or {density estimation}, within unlabeled data. In supervised learning, the machine is equipped with more power to predict the class or to infer the characteristic from the structured data.
The figures of merit for a learning machine include: (i) \emph{computational complexity} which measures the run-time efficiency of a learning algorithm; (ii) \emph{sample complexity} which determines the number of queries to a membership  made by the learning algorithm such that the hypothesis function is Probably Approximately Correct (PAC) \cite{Val84}; and (iii) \emph{model complexity} (otherwise called generalization error \cite{AML12}) which is defined as the discrepancy between the out-of-sample error and the in-sample error. Note that model complexity is closely related to sample complexity in the sense that a learning machine with large model complexity requires more samples to accurately approximate the target function, which results in high sample complexity. Current research trends include the reduction of computational complexity due a large volume data set (big data) as well as the high dimensional features of each data point, and how to balance model complexity with in-sample error such that the training data set can be trained well without the occurrence of overfitting.

\emph{Quantum Information Processing} (QIP) is an active field that studies the computational capability in quantum systems. In recent years, QIP has achieved significant breakthroughs \cite{NC00}: factorizing  large prime integers with an exponential speed-up \cite{Sho97} and searching an unstructured database with a quadratic speed-up \cite{Gro97} are two most famous examples. There are two features of QIP that result in dramatic improvement over classical information processing:
(1) The superposition principle: contrary to the classical bit, which takes discrete value either 0 or 1, a quantum bit (or \emph{qubit}) can be in any linear combination of two quantum states $|0\rangle$ and $|1\rangle$. The principle is a consequence of the fundamental property of quantum mechanics---the linearity of Schr\"{o}dinger's wave equation.
Therefore, the superposition principle allows the outcomes of parallel quantum computation to be stored in a single quantum state, which gives quantum machines more computing ability than classical devices.
(2) Entanglement: quantum entanglement is the most remarkable phenomenon in quantum theory. This resource plays a crucial role in numerous results, including  quantum Shannon theory \cite{DH13, HW10a, HW10b, HDW08}, quantum error-correcting codes \cite{BDH06, BDH14, WHB14, HYH11}, and so on.
These features make QIP a multidisciplinary research area with a broad range of promising applications. 

Owing to the successful achievements of QIP, researchers have begun to explore whether QIP can advance other subjects of classical computer science.
Consequently, the interdisciplinary area of quantum machine learning \cite{Wit14, SSP14b} has attracted substantial interest lately. The central problems are two-fold. The first kind of problem investigates how QIP can improve classical ML tasks by converting classical algorithms partially or totally to a quantum algorithm. More precisely, one studies how quantum machines can serve to accelerate the ML process to improve computational efficiency, or to reduce sample complexity by transforming classical training data into special sets of quantum states. We call this line of research \emph{Quantum Computational Learning} \cite{SG01, Ser01, SG04, AS05, ARR+03, ABG206, PL12, ABG13, WBL12, LMR13, RML13, LMR213, WKS14, Wan14, SSP14a, LGZ14, CCS14, SSP14d, WKS14b}. On the other hand, certain fundamental quantum problems, such as the inference of unknown quantum states or operations, or the hidden structure of the underlying quantum system, fits well into the setting of statistical learning theory. However, it requires certain generalisation of current theory of machine learning to accommodate the operator-valued inputs and/or outputs.  We term  this line of research \emph{Quantum Statistical Learning}\footnote{Our catalogue of quantum ML is different from the learning class $\mathcal{L}^\mathit{context}_\mathit{goal}$, where the subscript ``\emph{goal}'' refers to the learning goal and the superscript ``\emph{context}'' refers to the training data set and/or the learner's abilities, introduced by A\"{i}meur, Brassard, and Gambs \cite{ABG06}. According to the authors, $\mathcal{L}^\mathit{c}_\mathit{c}$ corresponds to pure classical ML tasks. When the learner can access to a quantum computer to accelerate the classical ML problems, it belongs to  the learning class $\mathcal{L}^\mathit{q}_\mathit{c}$.} \cite{ABG06, Aar07, Gam08, GK10, BCD+10, BDP+11, GLF+10, FGL+12, SCM+12, SGA14, LB14, BS13, AKK14, SSP14a, SSP14d, MBW10, BBA14, CHB+14, MW14}. 

Current achievements in Quantum Computational Learning come from quantum enhancement of the computation procedures such as optimization, inner product of big data and ability to compute classical functions in parallel.
For example, Servedio and Gortler \cite{SG01,Ser01,SG04,AS05} considered two standard learning models of Boolean functions: Angluin's \cite{Ang88} exact learning from membership queries, and Valiant's \cite{Val84} PAC learning from examples. By defining quantum extensions of the classical oracles to manipulate classical binary data, it was shown that the quantum oracles and classical machines are polynomially equivalent in terms of sample complexity.
Anguita \emph{et~al.} (2003) \cite{ARR+03} used the method of Durr and Hoyer  \cite{DH96} to perform the optimization process in support vector machine (SVM). A\"{i}meur, Brassard, and Gambs (2006) \cite{ABG206, ABG13} applied a modified Grover's algorithm \cite{Gro97} in clustering problems. Lloyd \emph{et~al.} (2013) \cite{LMR13, RML13, LMR213} introduced a quantum random access memory \cite{GLM08} to store classical data and proposed an efficient density matrix exponentiation method to improve the computational procedure of supervised, unsupervised and SVM algorithms. Additionally, Lloyd \emph{et al.} \cite{LGZ14} also provided quantum algorithms to execute topological analysis for big data. Pudenz and Lidar (2012) \cite{PL12} considered the verification  of software and applied adiabatic quantum computation methods to solve the quadratic binary optimization problem. Wiebe, Kappor, and Svore (2014) \cite{WKS14} (Microsoft Research) modified  Lloyd's approach and proposed a quantum nearest-neighbor algorithm. Surprisingly, they showed that the number of queries depends on the sparsity and maximum value of the training data rather than on the feature dimension. Wang (2014) \cite{Wan14} combined phase estimation and the dense Hamiltonian simulation technique to improve the ML performance in curve fitting.
Cross \emph{et al.} \cite{CCS14} considered the problem of learning parity functions in the presence of noise. They showed that the quantum oracle is computationally efficient than the classical counterpart.
Schuld \emph{et al.} \cite{SSP14d} presented a quantum pattern classification and discussed its advantages.
Recently, Wiebe \emph{et al.} \cite{WKS14b} successfully applied quantum computers to perform an important machine learning task---deep learning.
We refer the interested readers to Ref.~\cite[Table 1.1]{Wit14}, where Wittek provides a detailed comparisons of existing quantum machine learning algorithms.

On Quantum Statistical Learning, A\"{i}meur, Brassard, and Gambs \cite{ABG06} introduced the task of \emph{quantum clustering}, where the goal is to group similar quantum states (according to some fidelity measure) while putting dissimilar states in different clusters. 
In Ref. \cite{Gam08}, Gambs (2008) studied the task of \emph{quantum classification}, in which the training data set contains pure states from (classical) binary classes. By forming the statistical mixture states of each class, the Helstrom measurement forms a binary classifier which minimizes the training error.
Gu\c{t}\v{a} and Kotlowski (2010) \cite{GK10} researched the problem of classifying two unknown mixed states and used the technique of \emph{local asymptotic normality} to derive the optimal classifier.
Sen\'{i}s \emph{et al.} \cite{SCM+12, SGA14} also proposed the strategy to perform quantum state classification.
Nevertheless, the approaches developed by Gambs, Gu\c{t}\v{a} and Kotlowski, and Sen\'{i}s \emph{et al.}~are essentially \emph{quantum hypothesis testing} rather than quantum ML (since they do not consider the model complexity or sample complexity problems). In \cite{BCD+10}, Bisio \emph{et~al.} (2010) considered learning a unitary transformation as the storing-retrieving problem and proposed an algorithm for learning $U$ based on the quantum network. Further, Bisio \emph{et~al.}~(2011) \cite{BDP+11} generalised the previous work to quantum instruments.
Gross and Flammia \emph{et~al.} \cite{FGL+12, GLF+10} integrated the compressed sensing methods to quantum states tomography and proposed  an algorithm to practically learn row-rank quantum states.
In the latest work, Lu and Braunstein (2014) \cite{LB14} studied the quantum version of the decision tree (QDT). 
In their model, both the input variable $x$ and output label $y$ are represented as pure states $|x\rangle$, and $|y\rangle$.
The von Neumann entropy is used as the node splitting criterion to construct the quantum decision tree classifier.
Several works \cite{BS13, AKK14, SSP14a, SSP14d} has engaged in developing possible models of quantum neural network (QNN).
Another interesting topic is the hidden quantum Markov model (HQMM) \cite{MBW10, BBA14, CHB+14, MW14}, where the state of the system is described by the density operator and the transitions are determined by the completely positive trace-nonincreasing map. It has been shown that the HQMM can be implemented by open quantum systems with instantaneous feedback.

We summarise the current development of supervised quantum ML in Figure \ref{table:survey}. Note that the majority of previous works in quantum machine learning focused on computational aspects of a learning algorithm. The issue of sample complexity exhibited in original quantum learning setting, e.g.~state/process tomography, was rarely touched. Aaronson \cite{Aar07} pioneered the study of the learnability problems in the quantum regime, and derived upper bounds on the sample complexity of learning quantum states. In this work, we start from a machine learning point of view to formalize the problems of learning quantum measurements and quantum states as learning real-valued functions on Banach space. For learning an unknown quantum measurement, we apply a sequence of quantum states through the measurement apparatus and obtain the statistics of each measurement outcome. Our goal is to infer the most likely quantum measurement from the hypothesis set, which `behaves' like the target measurement on the collected data. In this paper, we mainly focus on learning an unknown  two-outcome measurement, which resembles a `yes-no' instrument. For multi-outcome measurements, the results can easily be generalised\footnote{In the scenario of learning multi-outcome measurements, each POVM element can be considered as a two-outcome POVM. Hence, the learnability of each POVM element can be derived by following the proposed paradigm. We note that this problem can be tackled by the \emph{multi-label learning} algorithms (also called \emph{multi-target prediction} or \emph{multivariate regression}.)
	}.
For learning quantum states, on the other hand, the training data set is the collection of two-outcome measurements and the associated statistics. The core problem now is to analyse whether the target quantum measurement is learnable and to characterise the performance of the learning tasks.

\begin{figure}
	\caption{Current Development of Quantum Machine Learning. `Quantum Computational Learning'  investigates how quantum machines can serve to accelerate the ML process to improve computational efficiency, or to reduce sample complexity by transforming classical training data into special sets of quantum states. In this line of research, both the input space $\mathcal{X}$ and output space $\mathcal{Y}$ are classical. On the other hand, `Quantum Statistical Learning' studies the inference of unknown quantum states, operations, or hidden structure in the quantum system. We term the quantum version of classical statistical/stochastic model as `Quantum Stochastic Model'.} \label{table:survey}
	\tikzstyle{block} =
	[
	rectangle
	, draw
	, text width=15.5em
	, text centered
	, rounded corners
	, minimum height=3em
	]
	
	\tikzstyle{block2} =
	[
	rectangle
	, draw
	, text width=10.5em
	, text centered
	, rounded corners
	, minimum height=3em
	]
	
	\tikzstyle{line} =
	[
	draw
	, -latex'
	]
	\begin{tikzpicture}       
	\matrix (mat) [matrix of nodes, nodes=block, column sep=-25mm, row sep=1cm]
	{ & & & & & & \\
		& & &\node [block] (root) {\textbf{Quantum Machine Learning}};  & & & \\
		&\node [block] (c1) {\textbf{Quantum\\Computational Learning}};  & & & &
		\node [block] (c2)  {\textbf{Quantum Statistical Learning}}; &     \\
		\node [block2] (c11) {\textbf{Computational Complexity}\\
			Angluin \cite{Ang88}\\
			A\"{i}meur \emph{et al.} \cite{ABG206, ABG13}\\
			Pudenz and Lidar \cite{PL12}\\
			Cross emph{et al.} \cite{CCS14}\\
			Lloyd \emph{et al.} \cite{WBL12, LMR13, RML13, LMR213, LGZ14}\\
			Wang \cite{Wan14}\\
			Wiebe \emph{et al.} \cite{WKS14, WKS14b}\\
			Schuld \emph{et al.} \cite{SSP14d}};
		&&
		\node [block2] (c12)   {\textbf{Sample Complexity}\\Servedio \emph{et al.} \cite{SG01, Ser01, SG04, AS05}};
		&&
		\node [block2] (c21) {\textbf{Quantum\\Stochastic Model}\\
			Bisio \emph{et al.} \cite{BCD+10,BDP+11}\\
			Gross \emph{et al.} \cite{FGL+12, GLF+10}\\
			Classificatioin: \cite{Gam08,GK10,SCM+12,SGA14}\\
			QDT: \cite{LB14}\\
			QNN: \cite{BS13, AKK14, SSP14a, SSP14d}\\
			QHMM: \cite{MBW10, BBA14, CHB+14, MW14}};
		&&
		\node [block2] (c22) {\textbf{Sample Complexity}\\Aaronson \cite{Aar07} and this work};\\
	};
	\path[line] (root.south)    --+(0,-0.5) -| node [pos=0.3, above] {} (c1.north);
	\path[line] (root.south)   -- +(0,-0.5) -| node [pos=0.3, above] {} (c2.north);
	\path[line] (c1.south) -- +(0,-0.5) -| node [pos=0.3, above] {} (c11.north);
	\path[line] (c1.south) -- +(0,-0.5) -| node [pos=0.3, above] {} (c12.north);
	\path[line] (c2.south)    -- +(0,-0.5) -| node [pos=0.3, above] {} (c21.north);
	\path[line] (c2.south)    -- +(0,-0.5) -| node [pos=0.3, above] {} (c22.north);
	\end{tikzpicture}
\end{figure}
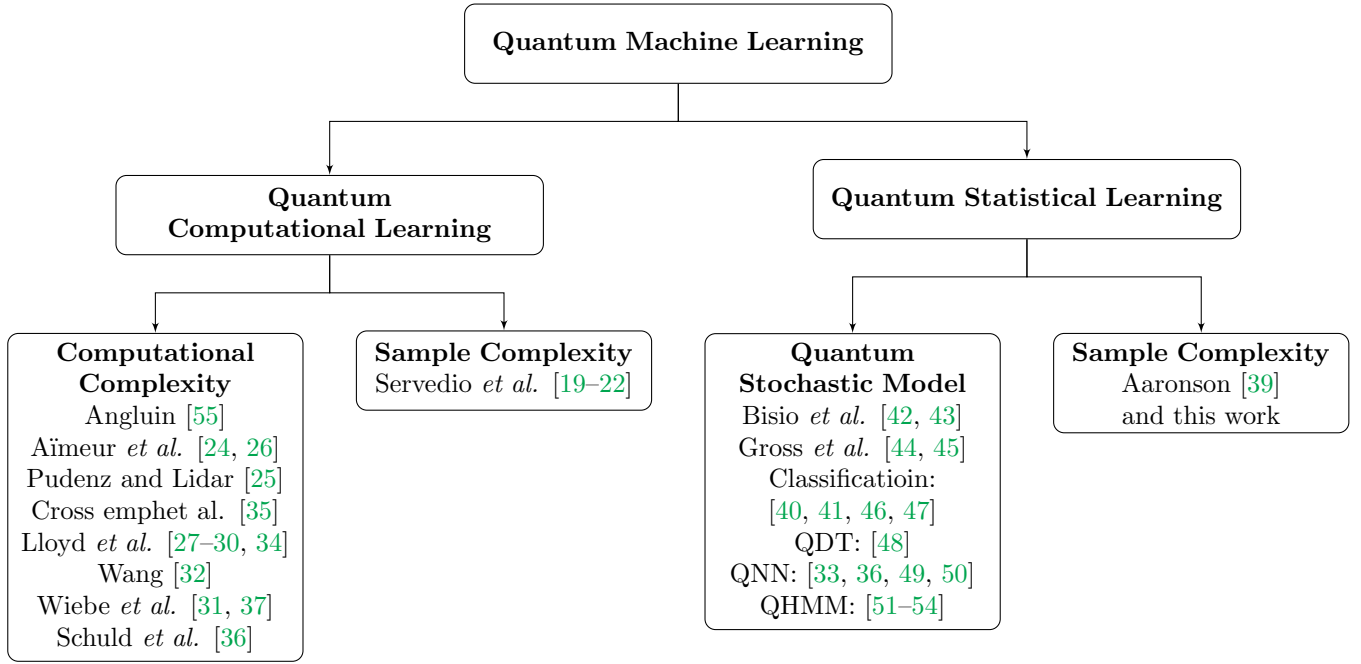

\subsection{Contributions of this work}
In this work, we answer the following two questions in quantum ML.

\textbf{How many quantum states are sufficient to learn a quantum measurement?}
Assume there is an unknown two-outcome quantum measurement device, and we can prepare a set of quantum states that are randomly drawn from an unknown distribution. Suppose that the outcome statistics of the set of quantum states are known. Can we infer the unknown quantum measurement from the quantum states at hand? How many samples of quantum states are needed for the learning machine to decide an optimal quantum measurement from the hypothesis set? Can the chosen candidate approximate the target measurement with the desired accuracy? These questions are typical sample complexity problems in statistical learning theory, and the answer lies in a proper quantification of the ``effective size" of the hypothesis set. In this paper, we propose a framework (see Section \ref{Framework}) to connect the problems of learning two-outcome measurements with the tasks of learning real-valued linear functional on quantum states. By exploiting Banach space theory and the \emph{noncommutative Khintchine inequalities} \cite{LP91} in Random Matrix Theory, we prove (Theorem \ref{theo:effect}) that the complexity measure---\emph{fat-shattering dimension}---is upper bounded by $O(d/\epsilon^2)$. Under the same framework, other complexity measures, such as covering numbers and Rademacher complexity, can be derived. As a result, the number of required sample states to learn an unknown quantum measurement is proportional to the dimension of the Hilbert space.

\textbf{How many quantum measurements are sufficient to learn a quantum state?}
Following the paradigm of learning quantum measurements, we can similarly formalize the problem of  learning an unknown state into its \emph{dual problem}. Unlike Aaronson \cite{Aar07}, we employ tools solely from statistical learning theory to show (Theorem \ref{theo:state}) that the fat-shattering dimension is $O(\log d/\epsilon^2)$ for learning a qudit state.  In addition, we also derive the covering number and the Rademacher complexity. Our results show that all three complexity measures are characterised by logarithmically proportional to the Hilbert dimension.

Lastly, by formulating the quantum learning problems into Bloch-sphere representation, we show that it is equivalent to a \emph{neural network}. Hence the classical ML algorithms can be practically applied to perform quantum ML tasks.
~\\~\\
There are several fields that may relate to or benefit from our work.
~\\~\\
\textbf{Quantum State/Measurement Tomography.}
Quantum state tomography is a difficult task in physics because the number of unknown parameters in a multi-partite quantum system grows exponentially.
Aaronson pointed out that quantum ML can serve as an alternative approach to quantum state tomography \cite{Aar07}. Surprisingly, learning an unknown target state within a given accuracy requires only the number of measurements that grows logarithmically with the dimension $d$. In this work, we push Aaronson's result one step further and consider application of machine learning framework to study quantum measurement tomography. To the best of our knowledge, there are very few results in this direction.
We hope that our result in learning quantum measurements will stimulate further investigation into this problem.

\bigskip
\textbf{Quantum State Discrimination.}
The goal of quantum state discrimination is to determine the identity of a state in an ensemble. Whenever states are not mutually orthogonal, they cannot be perfectly discriminated. Therefore, a possible way is ambiguous state discrimination with the goal of minimizing the error probability.
Given $\epsilon>0$ we show that the fat-shattering dimension guarantees that a set of quantum states can be discriminated into two subsets with the worst error probability no greater than $1/2-\epsilon$.
Following the same reasoning, the quantum states in the hypothesis set can be used to distinguish between two-outcome measurements. 

\bigskip
\textbf{Quantum Random Access Codes.}
The $(n,m,p)$-QRA coding stands for encoding an $n$-bit sequence into $m$-qubit  so that the receiver can recover any one of the bits with successful probability at least $p$. The information-theoretic inequalities of $n$ and $m$ provide an upper bound for the fat-shattering dimension of learning quantum states. Alternatively, we can use the complexity measure---pseudo dimension---to show that there exists no $(n,m,p)$-QRA coding scheme, with $n\geq2^{2m}$. The result coincides with the work of Hayashi \emph{et al.}~\cite{HIN+06}. See Section \ref{RQRAC} for further discussions.

\bigskip
The paper is organised as follows. In Section \ref{PSLT} we introduce the background of statistical learning theory (especially on supervised learning) and describe important complexity measures. In Section \ref{Framework}, we formalise a unified  framework to relate the problems of learning quantum measurements and learning quantum states with the learning real-valued functions. Based on the proposed approach, we derive learning quantum measurements and prove the main results in Section \ref{LQE}. In addition, we discuss the interpretations of the  to ambiguous set discrimination and also derive the covering numbers and the Rademacher complexity. In Section \ref{LQS}, we consider the problem of learning quantum states and discuss its relationship with QRA codes. In Section \ref{algorithms}, we formulate the learning problem into Bloch-sphere representation and discuss possible algorithms (e.g.~neural networks) to implement the quantum learning tasks. We conclude this paper in Section \ref{conclusion}. 

\begin{notation}
In this paper, we denote a Hilbert space by $\mathcal{H}$.
The trace of an operator $M$ on $\mathcal{H}$ is calculated as
\[
\Tr(M):=\sum_k e_k M e_k,
\]
where $\{e_k\}$ is any orthonormal basis on $\mathcal{H}$.
Let $\mathbb{M}_d$ denote the set of all self-adjoint operators on $\mathbb{C}^d$.
The Hilbert-Schmidt inner product on $\mathbb{M}_d$ can be defined as
$\langle A,B\rangle_\text{HS}:=\Tr(A B)$, where the subscript `HS' will be omitted when the context is clear.
For $p\in [1,\infty)$, we denote the Schatten $p$-norm of an operator $M$ as
\[
\|M\| _{p} := \bigg( \sum _{i\geq 1} |\lambda_i(M)|^p \bigg)^{1/p},
\]
where $\lambda_i(M)$ is the eigenvalue of $M$. We denote
$\|M\|_\infty := \sup_i |\lambda_i(M)|$ as the operator norm.
Clearly, $\|\cdot\|_1$ and $\|\cdot\|_2$ correspond to the trace
norm and Hilbert-Schmidt norm $\|\cdot\|_\text{HS}$ respectively. Slightly abusing the
notation, we also denote the conventional $\ell_p$ norm on
$\mathbb{R}^d$ by $\|\cdot\|_p$ for $p\in [1,\infty]$. We define the
unit ball associated with the Schatten norms as $S_p^d =\{
M\in\mathbb{M}_d :\, \|M\|_p \leq 1\}$.
The set of bounded operators on $\mathcal{H}$ is denoted as $\mathcal{B(H)}$, which is the set operators with finite Schatten $\infty$-norm.
Likewise, the set of operators with finite Schatten $1$-norm is called the set of trace class operators, $\mathcal{T(H)}$.

A \emph{quantum state} (also called \emph{density operators}) on the Hilbert space $\mathcal{H}$ is a positive semi-definite operator with unit trace.
We identify the \emph{state space} as the set of all quantum states on $\mathcal{H}$, i.e.~,
\[
\mathcal{Q}(\mathcal{H}):=\{\rho\in\mathcal{T(H)}:\, \rho\succeq 0,\, \Tr(\rho)=1\}.
\]
A positive operator-valued measure (POVM) on $\mathcal{H}$ is a finite set of positive semi-definite operators $\{\Pi_i\}_{i\in I}$ such that
\[
\sum_{i\in I} \Pi_i = \mathcal{I},
\]
where $\mathcal{I}$ denotes the identity operator on $\mathcal{H}$.
Each POVM element $\Pi_i$ is called a \emph{quantum effect}, which serves as an instrument to perform a yes-no measurement.
We denote the set of all effects as an \emph{effect space}:

\[ \label{effect_space}
\mathcal{E}(\mathcal{H}):=\{E\in\mathcal{B(H)}:\,\mathcal{O}\preceq E\preceq \mathcal{I}\}.
\]
All constants are denoted as $C$ or $c$ and are independent from other parameters. Their values may change from line to line. The
notation $A\lesssim B$ means there is a constant $c$ such that
$A\leq cB$ and $A\simeq B$ means both $A\lesssim B$ and $A\gtrsim
B$. We summarise all the notation in table \ref{table:notation} in Appendix \ref{notation}.
\end{notation}

\section{Background of Statistical Learning Theory} \label{PSLT}

The starting point of this section is the mathematical formalism of the \emph{supervised machine learning}. We describe the effectiveness of a learning machine and examine the number of samples required to produce an almost optimal function with an error rate below the desired accuracy. As will be shown later, the bound of the sample complexity is closely related to the measures of complexity which characterise the ``size" of a function class.

\subsection{Supervised Machine Learning} \label{FSLT}

Generally speaking, supervised learning is a ML task that infers a function (or a learning model) by observing the data and the response to the data.
In this work, we focus on the definitions of agnostic PAC learnability and sample complexity for supervised machine learning.
For more comprehensive introduction to ML,
we refer the readers to literature such as Refs.~\cite{DGL96, Vap98, AB99, BBL03, Men04, HTF11, SB14}.

Consider a probability space
$(\mathcal{Z},\mu)$, where $\mathcal{Z}:=\mathcal{X}\times\mathcal{Y}$ with $\mathcal{X}$ (called the \emph{input space}) a measurable space and $\mathcal{Y}$ (called the \emph{output space}) a closed subset of real line $\mathbb{R}$. The probability distribution $\mu$ over $\mathcal{Z}$ is assumed to be fixed but known only through the \emph{training data set}, i.e.~$\mathsf{Z}_n=\left\{(X_1,Y_1),\ldots,(X_n,Y_n)\right\}\in\mathcal{Z}^n$ sampled independently and identically according to the measure $\mu$.
Supervised learning aims to construct a function $f:\mathcal{X}\rightarrow\mathcal{Y}$ which approximates the functional relationship between the input variable $X\in\mathcal{X}$ and the output variable $Y\in\mathcal{Y}$ from the observed training data set.
To evaluate the performance of the approximation, we define the loss function as a measurable map $\ell_f:\mathcal{Z}\rightarrow[0,+\infty)$ and the \emph{expected risk} (also called the \emph{out-of-sample error}):
\[
L(f)=\mathbb{E}_{\mu}\ell_f(X,Y).
\]
The loss function is usually taken as the absolute error or square error, i.e.
\[
\ell_f (X,Y)=|f(X)-Y|\quad \text{or} \quad \ell_f (X,Y)=(f(X)-Y)^2.
\]
For convenience, we only consider the square error in this work.
Other loss functions that satisfy the Lipschitz condition can be easily generalised\footnote{
	A loss function $\ell_f:\mathcal{Z}\rightarrow (0,\infty)$ is a Lipschitz function if it satisfies the Lipschitz condition
	\[
	|\ell_f(X,Y)-\ell_g(X,Y)|\leq \mathsf{L}|f(X)-g(X)|
	\]
	for all possible $(X,Y)\in\mathcal{Z}$ and the quantity $\mathsf{L}\in\mathbb{R}$ is called the Lipschitz constant.
	Denote by $\ell_{\mathcal{F}}$ the set $\{\ell_f:f\in\mathcal{F}\}$.
	Then the complexity measures (e.g.~the covering number and Rademacher complexity) of the class $\ell_{\mathcal{F}}$ are different from that of the hypothesis set $\mathcal{F}$ by the Lipschitz constant $\mathsf{L}$ \cite{Nat93,BM02}, i.e.~
	\[
	\mathcal{N}_p (\epsilon,\ell_{\mathcal{F}},m)\leq\mathcal{N}_p (\epsilon/\mathsf{L},\mathcal{F},m)\quad \text{for } p\geq 1,\,m\in\mathbb{N}
	\]
	and
	\[
	\mathcal{R}_n(\ell_{\mathcal{F}})\leq \mathsf{L}\mathcal{R}_n(\mathcal{F}).
	\]
	Therefore, by homogeneity we may assume the loss function is the absolute error with $L=1$ or the square error $L=2$ for deriving the sample complexity problems.
	}.
	
Since we are interested in minimising the expected risk, hence the \emph{target function} (or \emph{Bayes function}) as $t(x)=\mathbb{E}[Y|X=x]$ can be defined to achieve the minimum expected risk (called the \emph{Bayes risk}), i.e.~
\begin{eqnarray} \label{eq:min}
L_\text{Bayes}:=L(t)=\inf_f L(f),
\end{eqnarray}
where the infimum is taken over all possible measurable functions from $\mathcal{X}$ to $\mathcal{Y}$.
When $y$ is a deterministic function of $X$, then $Y=t(X)$ almost surely and $L(t)=0$.

The goal of the learner is to identify the target function $t$ from a
collection of functions $\mathcal{F}$, called the \emph{hypothesis
set}\footnote{Note that we use the term `hypothesis set' and `function class' interchangeably throughout the paper.}, which is a set of real-valued functions defined on the input
space $\mathcal{X}$. A \emph{learning algorithm} $\mathsf{A}$ for hypothesis set $\mathcal{F}$ is a mapping that assigns to
every training data $\mathsf{Z}_n$ some candidate function $\mathsf{A}(\mathsf{Z}_n)\in\mathcal{F}$, i.e.~
\[
\mathsf{A}:\cup_{n=1}^{\infty}\mathsf{Z}^n\rightarrow \mathcal{F}.
\]
The effectiveness of the learning
algorithm is measured by the number of data required to produce an
almost optimal function in the sense of Eq.~\eqref{eq:min}.
Therefore, we introduce one of the most fundamental concepts in supervised machine learning---\emph{Agnostic Probably Approximately Correct (PAC)} learning model \cite{Val84, Hau92}:
\begin{defn}[Agnostic PAC Learnability \cite{SB14}, Def.~3.3] \label{defn:PAC}
	A hypothesis set $\mathcal{F}$ is agnostic PAC learnable if there exist a function $m_\mathcal{F}:\mathbb{R}\times\mathbb{R}\rightarrow \mathbb{N}$ and a learning algorithm with the following property: For every $\epsilon,\delta\in(0,1)$ and for every distribution $\mu$ over $\mathcal{Z}$, when running the learning algorithm on $n\geq m_\mathcal{F}(\epsilon,\delta)$ samples generated by $\mu$, the algorithm returns a hypothesis $\hat{f}$ such that, with probability of at least $1-\delta$ (over the choice of the $n$ training examples),
	\[
	L(\hat{f})\leq \inf_{f\in\mathcal{F}} L(f)+\epsilon.
	\]
\end{defn}

However, the expected risk
$L(f)=\mathbb{E}_\mu[\ell_f(X,Y)]$ cannot be calculated since $\mu$ is unknown. We can only evaluate the agreement
of a candidate function over the training data set, which is called the
\emph{empirical risk} (also called the \emph{in-sample error}):
\[ \label{eq:ER}
\widehat{L}_n(f)=\frac1n \sum_{i=1}^n \ell_f(X_i,Y_i).
\]
For example, one of the most well-known learning algorithms is the Empirical Risk Minimization (ERM) principle \cite{Vap98} that assigns a function $f_n\in\mathcal{F}$ to each training data set which is ``almost optimal" on the data, i.e.~
\begin{eqnarray} \label{eq:ERM}
f_n=\arg \min _{f\in\mathcal{F}} \widehat{L}_n(f).
\end{eqnarray}

One way to evaluate the performance of the learning algorithm is to relate the risk $L(f_n)$ to the empirical risk $\widehat{L}_n(f_n)$. Following the reasoning of agnostic PAC model, our goal is hence to estimate the \emph{generalisation error} $\epsilon$:
\[ \label{eq:empirical_bound}
L(f_n)\leq \widehat{L}_n(f_n) + \epsilon(n,\mathcal{F}). 
\]
For any algorithm that outputs a $f_n\in\mathcal{F}$, we have
\[ \label{eq:sup_bound}
L(f_n)- \widehat{L}_n(f_n)\leq \sup_{f\in\mathcal{F}} \{ L(f)-\widehat{L}_n(f) \},
\]
which leads to the definition of \emph{uniform Glivenko-Cantelli class} (uGC class).
\begin{defn}
We say that the hypothesis set $\mathcal{F}$ is a uniform Glivenko-Cantelli class if for every $\epsilon>0$,
\[
\lim_{n\rightarrow\infty} \sup_\mu \Pr \left\{ \sup_{f\in\mathcal{F}} \left|L(f)-\widehat{L}_n(f)\right|\geq \epsilon \right\}=0.
\]
\end{defn}
The uniformity is with respect to all members of $\mathcal{F}$ and over
all possible probability measures $\mu$ on the domain $\mathcal{Z}$.
In addition to the conditions of the learnability, we also consider the bound on the rate of uniform convergence. For every $0<\epsilon,\delta<1$, let $m_\mathcal{F}(\epsilon,\delta)$ be the first integer such that for every $n\geq m_\mathcal{F}(\epsilon,\delta)$ and any probability measure $\mu$,
\begin{eqnarray} \label{eq:sub_bound2}
\Pr \left\{ \sup_{f\in\mathcal{F}} \left|L(f)-\widehat{L}_n(f)\right|\geq \epsilon \right\}\leq \delta.
\end{eqnarray}
The quantity $m_\mathcal{F}(\epsilon,\delta)$ satisfied Eq.~\eqref{eq:sub_bound2} is called the (Glivenko-Cantelli) \emph{sample complexity} of the hypothesis set  $\mathcal{F}$ with accuracy $\epsilon$ and confidence $\delta$. The sample complexity encapsulates the number of samples required to learn a set of functions.

Vapnik studied the relation between the uGC class and learnability \cite{Vap92, Vap95, Vap98} and showed that if a hypothesis set
$\mathcal{F}$ is a uGC class, then it is sufficient for the agnostic PAC
learnability\footnote{Agnostic PAC learnable is also called \emph{learnable with ERM}, or we can say that the ERM algorithm is \emph{consistent}. Recent works consider the \emph{stability} issues of the learning algorithm as one of the criterion of learnability. However, in this paper we do not deal with issues of stability and refer interested readers to Refs.~\cite{SSS+10, VRP13} and the references therein.}.

\begin{theo}
	[Uniform Convergence \protect{\cite[Corollary 4.4]{SB14}}] \label{theo:uniform}
	A training data set $\mathsf{Z}_n$ is called $\epsilon$-representative (with respect to domain $\mathcal{Z}$, hypothesis set $\mathcal{F}$, loss function $\ell$, and distribution $\mu$) if
	\[
	\forall f\in\mathcal{F},\quad \left|\widehat{L}_n(f)-L(f)\right|\leq\epsilon.
	\]
	Then, for every $\epsilon,\delta\in(0,1)$ and every probability distribution $\mu$ over $\mathcal{Z}$, a uGC class $\mathcal{F}$ that guarantees an $\epsilon/2$-representative set with probability of at least $1-\delta$ is agnostic PAC learnable.
	Furthermore, the ERM algorithm is an agnostic PAC learner for $\mathcal{F}$.
\end{theo}
As a result, we consider the generalisation error $\epsilon(n,\mathcal{F})$ and the sample complexity $m_\mathcal{F}(\epsilon,\delta)$ of the hypothesis set $\mathcal{F}$ as the performance criterion to investigate whether the underlying learning problem is agnostic PAC learnability.

In summary, the fundamental problems in ML are two-fold. The first is under what conditions the machine is agnostic PAC learnable.
Secondly, the sample complexity determines the rate of the uniform convergence and the information-theoretic efficiency of the hypothesis set $\mathcal{F}$.
In the next subsection, several complexity measures are introduced to characterise the ``richness" or ``effective size'' of the hypothesis set.
In Appendix \ref{SC}, we show that the sample complexity can be further expressed in terms of the complexity measures.


\subsection{Measures of Sample Complexity} \label{MC}
As discussed before, we are interested in the parameters which effectively measure the size of a given hypothesis set. There are some well-known measures of (information) complexity\footnote{The complexity measures introduced in this section and the generalisation bounds derived in Section \ref{SC} are information-theoretic in the sense that the learning algorithms are based on the agnostic PAC model regardless of the computational resources.} of the function class: \emph{combinatorial parameters}, \emph{covering numbers}, and \emph{Rademacher complexity}.

The first combinatorial parameter---\emph{Vapnik-Chervonenkis (VC) dimension}---was introduced by Vapnik and Chervonenkis \cite{VC71} for learning Boolean functions.

\begin{defn}[VC Dimension]
Let $\mathcal{F}$ be a set of $\{0,1\}$-valued functions on a domain $\mathcal{X}$. We say that $\mathcal{F}$ shatters a set $\{x_1,\ldots, x_n\}\subseteq \mathcal{X}$ if for every subset $B\subseteq \{1,\ldots,n\}$ there exists a function $f_B\in\mathcal{F}$ for which $f_B(x_i)=1$ if $i\in B$, and $f_B(x_i)=0$ if $i\notin B$. Let
\[
\text{VCdim}(\mathcal{F})=\sup \left\{ |\mathcal{S}| :\, \mathcal{S}\subseteq \mathcal{X},\,\mathcal{S} \text{ is shattered by } \mathcal{F}\right\}.
\]
The \emph{VC dimension} of $\mathcal{F}$ (on the domain $\mathcal{X}$) is denoted as VCdim$(\mathcal{F})$. 
\end{defn}

Pollard \cite{Pol84} generalised the concept of VC dimension and introduced the \emph{pseudo dimension} to quantify the sample complexity of a real-valued function class. The parameterised version of Pollard's pseudo-dimension is the \emph{scale-sensitive dimension} (also called the \emph{fat-shattering dimension}) introduced by Kearns and Schapire \cite{KS94}.

\begin{defn}[Pseudo Dimension] \label{defn_pseudo}
Let $\mathcal{F}$ be a set of real-valued functions on a domain $\mathcal{X}$.
We say a set $\mathcal{S}=\{x_1,\ldots,x_n\}\subseteq \mathcal{X}$ is pseudo-shattered by $\mathcal{F}$ if there exists a set $\{\alpha_i\}_{i=1}^n$ such that for every $B\subseteq \{1,\ldots,n\}$ there is some function $f_B\in\mathcal{F}$ for which $f_B(x_i)\geq \alpha_i$ if $i\in B$, and $f_B(x_i)< \alpha_i$ if $i\notin B$. Define the \emph{pseudo dimension} of $\mathcal{F}$ as
\[
\text{Pdim}(\mathcal{F})=\sup \left\{ |\mathcal{S}| :\, \mathcal{S}\subseteq \mathcal{X},\,\mathcal{S} \text{ is pseudo-shattered by } \mathcal{F}\right\}.
\]
$f_B$ is called the shattering function of the set $\mathcal{S}$. 
\end{defn}
There is a desirable property of the pseudo dimension that will be useful in our main theorems.

\begin{theo}[Pollard \cite{Pol84}] \label{theo:pseudo}
 \hfill
 \begin{itemize}
\item[(i)] If $\mathcal{F}$ is a vector space of real-valued functions then Pdim$(\mathcal{F})=$ dim$(\mathcal{F})$.
\item[(ii)] If $\mathcal{F}$ is a subset of a vector space $\mathcal{F'}$ of real-valued functions then Pdim$(\mathcal{F})\leq$ dim$(\mathcal{F'})$.
\end{itemize}
\end{theo}

\begin{defn}[Fat-Shattering Dimension] \label{defn:fat}
Let $\mathcal{F}$ be a set of real-valued functions on a domain $\mathcal{X}$.
For every $\epsilon>0$, a set $\mathcal{S}=\{x_1,\ldots,x_n\}\subseteq \mathcal{X}$ is said to be $\epsilon$-shattered by the $\mathcal{F}$ if there exists a set $\{\alpha_i\}_{i=1}^n\subset \mathbb{R}$ such that for every $B\subseteq \{1,\ldots,n\}$ there is some function $f_B\in\mathcal{F}$ for which $f_B(x_i)\geq \alpha_i+\epsilon$ if $i\in B$, and $f_B(x_i)< \alpha_i-\epsilon$ if $i\notin B$. Define the \emph{fat-shattering dimension} of $\mathcal{F}$ on the domain $\mathcal{X}$ as
\[
\text{fat}_\mathcal{F}(\epsilon,\mathcal{X})=\sup \left\{ |\mathcal{S}| :\, \mathcal{S}\subseteq \mathcal{X},\,\mathcal{S} \text{ is $\epsilon$-shattered by } \mathcal{F}\right\}.
\]
$f_B$ is called the shattering function of the set $B$ and the set $\{\alpha_i\}_{i=1}^n$ is called a witness to the $\epsilon$-shattering. When the underlying space is clear, we denote it by fat$_\mathcal{F}(\epsilon)$.
If the witness set $\{\alpha_i\}$ are all equal to a constant, we call it as the \emph{level fat-shattering dimension}, $\underline{\text{fat}}_\mathcal{F}(\epsilon)$.
\end{defn}

In Ref. \cite{AB99}, a relationship between the fat-shattering dimension and the pseudo-dimension can be given.
\begin{theo}[Anthony and Bartlett \cite{AB99}] \label{theo:AB99}
Let $\mathcal{F}$ be a set of real-valued functions. Then:
\begin{itemize}
\item[(i)] For all $\epsilon>0$, fat$_\mathcal{F}(\epsilon)\leq \text{Pdim}(\mathcal{F})$.
\item[(ii)] If a finite set $\mathcal{S}$ is pseudo-shattered then there is $\epsilon_0$ such that for all $\epsilon>\epsilon_0$, $\mathcal{S}$ is $\epsilon$-shattered.
\item[(iii)] The function fat$_\mathcal{F}(\epsilon)$ is non-increasing with $\epsilon$.
\item[(iv)] Pdim$(\mathcal{F})=\lim_{\epsilon\downarrow0}\text{fat}_\mathcal{F}(\epsilon)$ (where both sides may be infinite).
\end{itemize}
\end{theo}
Note that it is possible for the pseudo-dimension to be infinite, even when the fat-shattering dimension is finite for all positive $\epsilon$.

In addition to the combinatorial parameters bounding the sample complexity, there are other quantities called \emph{covering number} which measure the size of the function class by the finite approximating set. The concept of covering number dates back to Kolmogorov \emph{et al.} \cite{KV61} and has been used in many areas of mathematics.

\begin{defn}[Covering Number] \label{defn:cn}
Let $(Y,d)$ be a metric space and let $\mathcal{F}\subset Y$. For every $\epsilon>0$, the set $\{y_1,\ldots,y_n\}$ is called an $\epsilon$-cover of $\mathcal{F}$ if every $f\in\mathcal{F}$ has some $y_i$ such that $d(f,y_i)<\epsilon$. The covering number $\mathcal{N}(\epsilon,\mathcal{F},d)$ is the minimum cardinality of a $\epsilon$-covering set for $\mathcal{F}$ with respect to the metric $d$.
\end{defn}

To characterise the size of the function class $\mathcal{F}$ in machine learning, we investigate the metrics endowed by the samples; for every sample $\{x_1,\ldots,x_n\}\in\mathcal{X}$, let $\mu_n=n^{-1}\sum_{i=1}^n \delta_{x_i}$ be the empirical measure supported on that sample. For $1\leq p<\infty$ and a function $f$, denote $\|f\|_{L_p (\mu_n)}=\left(n^{-1}\sum_{i=1}^n |f(x_i)|^p\right)^{1/p}$ and $\|f\|_\infty=\max_{1\leq i\leq n}|f(x_i)|$. Then, $\mathcal{N}\left(\epsilon,\mathcal{F},L_p(\mu_n)\right)$ is the covering number of $\mathcal{F}$ at scale $\epsilon$ with respect to the $L_p(\mu_n)$ norm.

\begin{defn}[Entropy Number]
For every class $\mathcal{F}$, $1\leq p\leq \infty$ and $\epsilon>0$, let
\[
\mathcal{N}_p(\epsilon,\mathcal{F},n)=\sup_{\mu_n} \mathcal{N}\left( \epsilon,\mathcal{F},L_p(\mu_n)\right),
\]
and
\[
\mathcal{N}_p(\epsilon,\mathcal{F})=\sup_n\sup_{\mu_n} \mathcal{N}\left( \epsilon,\mathcal{F},L_p(\mu_n)\right).
\]
We call $\log \mathcal{N}_p(\epsilon,\mathcal{F},n)$ the \emph{entropy number} of $\mathcal{F}$ with respect to $L_p(\mu_n)$ and $\log \mathcal{N}_p(\epsilon,\mathcal{F})$ the \emph{uniform entropy number}.
\end{defn}

The significance of the uniform measures of complexity (i.e.~uniform entropy number and combinatorial parameters) lies in that they can characterise the uGC class. However, the bounds are loose. Bartlett and Mendelson \cite{BM02} considered the techniques of concentration of measures for empirical processes and proposed a random average quantity---\emph{Rademacher complexity}, which capture the size of the uGC class more directly and leads to sharp complexity bounds.

\begin{defn}[Rademacher Complexity\footnotemark \cite{KS94, BM02, Men02}]\label{def_RCom}
Let $\mu$ be a probability measure on $\mathcal{X}$ and $\mathcal{F}$ be a set of uniformly bounded functions on $\mathcal{X}$. For every positive integer $n$, define
\[
\mathcal{R}_n(\mathcal{F})=\mathbb{E} \sup_{f\in\mathcal{F}} \frac1{\sqrt{n}}\left| \sum_{i=1}^n\gamma_i f(x_i)\right|,
\]
where $\{x_i\}_{i=1}^n$ are independent random variables distributed according to $\mu$ and $\{\gamma_i\}_{i=1}^n$ independently takes values in $\{-1,+1\}$ with equal probability (which are also independent of $\{x_i\}_{i=1}^n$). The quantity $\mathcal{R}_n(\mathcal{F})$ is called the Rademacher complexity associated with the class $\mathcal{F}$.
\end{defn}

We remark that the complexity measures can be related among each other \cite{Dud67, Sud71,MV03}:
\[
\text{fat}_\mathcal{F}(\epsilon)\lesssim \log\mathcal{N}_2(\epsilon,\mathcal{F},n)\lesssim \frac{\mathcal{R}_n^2(\mathcal{F})}{\epsilon^2} \lesssim \text{fat}_\mathcal{F}(\epsilon)\cdot\log\left(\frac{1}{\epsilon}\right).
\]

\footnotetext{
	Some authors define the Rademacher complexity with the normalisation term as $n$ rather than $\sqrt{n}$. Here we follow the notation used in Ref.~\cite{Men02}, which is more convenient to bound the sample complexity (e.g.~Eq.~\eqref{eq:sc_ra}).
	}

To sum up the results we have presented so far, the complexity measures, such as the combinatorial parameters (e.g.~VC dimension and fat-shattering dimension), covering numbers and the Rademacher complexity of the hypothesis set control the rate of uniform convergence. By computing those quantities of the given hypothesis set and according to Eqs. \eqref{eq:sc_VC}, \eqref{eq:sc_fat}, \eqref{eq:sc_cn}  and \eqref{eq:sc_ra}, we can estimate the bounds on the sample complexity of the learning problems.


\section{The Framework for Learning Quantum Measurements and Quantum States}
\label{Framework}

In this section, we unify the two quantum learning problems at hand into learning linear functionals. In Section \ref{justification}, we justify the proposed quantum learning model in practical situations.

\subsection{Quantum Learning Problems as Linear Functional on Matrices} \label{relation}

Recall that a physical theory aims to predict events observed in the experiments by describing three types of apparatus: preparation, transformation, and measurement. The preparation process of a system can be embodied by a state, while an \emph{effect} is a measurement that produces either `yes' or `no' outcomes in order to observe the physical experiment. However, according to the statistical nature of Quantum Theory, only probabilities of the occurrence can be predicted (counting multiple measurements).
More precisely, assume that a system is prepared in the state $\rho\in\mathcal{Q(H)}$. Then the outcome of every two-outcome measurement $E\in\mathcal{E(H)}$ takes the form of the probability distribution:
\[
f_E(\rho)=\Tr(E\rho)=\left\langle E,\rho \right\rangle\in[0,1].
\]
Note that it is a linear functional on the state space, i.e.~$f_E:\mathcal{Q(H)}\rightarrow \mathbb{R}$. In the theory of learning, such $[0,1]$-valued functions are called \emph{probabilistic concepts} \cite{KS94}.

The following proposition establishes the one-to-one correspondence between $f_E\leftrightarrow E$.

\begin{prop}[The Correspondence between Two-Outcome Measurement and Linear Functional]{\protect{\cite[Prop.~2.30]{Hei12}}}  \label{prop:effect}
	Given a Hilbert space $\mathcal{H}$, let $f_E$ be an effect, i.e.~a linear map from $\mathcal{Q(H)}$ to the interval $[0,1]$.
	Then there exists a bounded operator $E \in\mathcal{E(H)}$ such that
	\[
	f_E(\rho) = \Tr(E\rho) =\left\langle E,\rho \right\rangle \quad \forall \rho\in\mathcal{Q(H)}.
	\]
	Furthermore, the operator $E$ is unique in the following sense. Let $E_1,\,E_2\in\mathcal{E(H)}$. If $\langle \varphi,E_1\varphi\rangle=\langle \varphi,E_2 \varphi\rangle$ for every $|\varphi\rangle\in\mathcal{H}$, then $E_1=E_2$.
\end{prop}
The proposition states that every two-outcome measurement can be identified as a linear functional on the state space.
Consequently, the problem of learning an unknown (two-outcome) quantum measurement is equivalent to learning a real-valued linear functional on quantum states. Here and subsequently, we call an effect to represent either the linear functionals on $\mathcal{Q(H)}$ or the two-outcome measurement $E\in\mathcal{E(H)}$.

Conversely, if the measurement apparatus is chosen as some $E\in\mathcal{E(H)}$, then the measurement outcome of every state $\rho$ is distributed as
\[
f_\rho(E)=\Tr(E\rho)=\left\langle E,\rho \right\rangle\in[0,1]\quad \forall \rho\in\mathcal{Q(H)}.
\]
Therefore, we take the state space as the set of linear functionals on the effect space by the following proposition:

\begin{prop}
	[The Correspondence between Quantum State and Linear Functional on Effect Space]{ \cite{Bus03}} \label{prop:state}
	Given a Hilbert space $\mathcal{H}$, let $f_\rho$ be probability measure on $\mathcal{E(H)}$.
	Then there exists a quantum state $\rho\in\mathcal{Q(H)}$ such that
	\[
	f_\rho(E) = \Tr(E\rho) =\left\langle E,\rho \right\rangle \quad \forall E\in\mathcal{E(H)}.
	\]
	Furthermore, different $\rho_1,\,\rho_2\in\mathcal{Q(H)}$ determines different probability measures, i.e.~there exists an operator $E\in\mathcal{E(H)}$ such that $\Tr(E\rho_1)\neq\Tr(E\rho_2)$.
\end{prop}
Similarly, according to the one-to-one correspondence between $\rho\leftrightarrow f_\rho$, learning an unknown quantum state coincides with learning a real-valued linear functional on the effect space.

\subsection{Learning Linear Functionals on Banach Space} \label{Banach}

In the previous section, we establish the relationship between quantum measurements/states and linear functional on matrices. 
By the duality theorem  (see Theorem \ref{theo:duality} below), the two quantum learning problems can be unified into the problem of learning the membership in a Banach space.
{Furthermore, the real-valued function that associates with the target quantity in the Banach space is isomorphic to the linear functional on the input space, i.e.~an element in the dual space of the input space.}
For example, assume the input space is the unit ball of the Schatten $p$-class, i.e.~$\mathcal{X}=S_p^d$.
Then the hypothesis set can be represented as the linear functionals that are polar\footnote{In convex analysis, a \emph{convex body} $K\subset \mathbb{R}^n$ is a convex compact set with nonempty interior. The \emph{gauge} of a convex body $K$, also known as the \emph{Minkowski functional}, is defined by $\|x\|_K:=\inf\{t\geq 0:\, x\in tK\}$. If $K$ is symmetric with respect to the origin ($-K=K$), then $K$ is a unit ball associated with the norm $\|\cdot\|_K$ and the inner product $\langle\cdot,\cdot\rangle$.
	We define the \emph{polar} of $K$ as
	\[
	K^\circ=\left\{ x\in\mathbb{R}^n: \, \sup_{k\in K}\langle k, x\rangle \leq 1\right\}.
	\]
In the symmetric case, $K^\circ$ is the unit ball of the dual space of $(\mathbb{R}^n,\|\cdot\|_K)$. Here, $S_1^d$ is a unit ball of Schatten 1-class and $S_\infty^d$ is a unit ball of Schatten $\infty$-class. Considering the Hilbert-Schmidt inner product,	$S_1^d$ and $S_\infty^d$ are polar to each other.} to $S_p^d$, i.e. for all $x\in S_p^d$ and $1/p+1/q=1$,
\[
\mathcal{F}=\left\{x\mapsto \left\langle E,x\right\rangle : E\in S_q^d\right\}=\left(S_p^d\right)^\circ.
\]
Under this duality formalism, the problems of estimating the complexity measures of the subset in a Banach space can be transformed into the following question: Whether a set of linear functionals is agnostic PAC learnable?

\begin{theo}[Duality of Bounded Operator and Trace class]{\cite[Theorems~19.1 and 19.2]{Con99}} \label{theo:duality}
Fix a Hilbert space $\mathcal{H}$. The map $E\mapsto f_E$ is an isometric isomorphism from the space of bounded operators, $\mathcal{B(H)}$, to the dual space of the set of trace classes operators, $\mathcal{T(H)^*}$.
	Conversely,  the map $\rho\mapsto f_{\rho}$ is an isometric isomorphism from $\mathcal{T(H)}$ to $\mathcal{B(H)^*}$.
\end{theo}

Mendelson and Schechtman \cite{MS04} first investigated the fat-shattering dimension of sets of linear functionals on Banach space and proposed the following useful result.

\begin{lemm}[Mendelson and Schechtman \cite{MS04}] \label{lemm:Men04}
	The set $\mathcal{S}=\{x_1,\ldots,x_n\}\subset B_X$ is $\epsilon$-shattered by $B_{X^*}$ if and only if $\{x_i\}_{i=1}^n$ are linearly independent and for every $a_1,\ldots,a_n\in\mathbb{R}$,
	\[
	\epsilon\sum_{i=1}^n|a_i|\leq \left\|\sum_{i=1}^n a_i x_i\right\|_{X},
	\]
	where $B_X$ is the unit ball of some Banach space $X$ and $B_{X^*}$ is its dual unit ball.
\end{lemm}

By restricting the values of the set $\{a_i\}_{i=1}^n$ to $\{+1,-1\}$,
the core idea of Lemma \ref{lemm:Men04} is to calculate the \emph{Rademacher series} on the Banach space, where the $n$ points Rademacher series on $\mathcal{X}$ is defined as $\sum_{i=1}^n \gamma_i x_i$,
where $\{\gamma_i\}_{i=1}^n$ are the symmetric $\{+1,-1\}$-valued random variables.
Additionally, with the following duality formula for the Schatten $p$-norm, we can estimate the range of the linear functional, which will helpful to further derive the complexity measures.

\begin{theo}
	[Duality Formula for $\|A\|_p$]{\cite[Theorem~7.1]{Car10}} \label{theo:Holder}
	For all $p\geq1$, define $q$ by $1/q+1/p=1$. Then for all $A\in\mathbb{M}_d$,
	\[
	\|A\|_p=\sup_{B\in\mathbb{M}_d} \left\{\Tr (B A):\|B\|_q =1\right\}.
	\]
\end{theo}

The techniques from Mendelson and Schechtman (Lemma \ref{lemm:Men04}) and the duality formula (Theorem \ref{theo:Holder}) can be used to upper bound the fat-shattering dimension and the Rademacher complexity via the Rademacher series. What remains is to compute the Rademacher series on the Banach space for both complexity measures, and we leave the details to Sections \ref{LQE} and \ref{LQS}.

\subsection{The Justification of the Quantum Learning Model} \label{justification}

Before proceeding to derive the complexity measures, we first address two practical issues that may arise in our quantum learning setting:
(1) Only the `yes' (`1') or `no' (`0') outcome can be observed rather than the outcome statistics\footnote{The situation can also occur when only one measurement is performed.}.
(2) The measurement apparatus is not perfect (e.g.~there are measurement errors in the training data set). However, we will show that the sample complexities of the two scenarios remain the same (up to a Lipschitz constant).

~\\
\textbf{The output space consists of binary measurement outcomes rather than measurement statistics.}
In this case, 
the training sample $(X_i,Y_i)$ equals to $(X_i,1)$ with probability $\Tr(\Pi X_i)$, and $(X_i,0)$ with probability $1-\Tr(\Pi X_i)$.
We show that the covering number remains the same as the training sample $(X_i, \Tr(\Pi X_i))$ considered in the quantum machine learning setting. Other complexity measures easily follow by the same argument.
Assume the underlying loss function $\ell_f$ satisfies the Lipschitz condition, i.e.~there exists $L>0$ such that
\begin{eqnarray} \label{eq:Lipschitz}
\left|\ell_f(X,Y)-\ell_g(X,Y)\right|\leq \mathsf{L}\left|f(X)-g(X)\right|.
\end{eqnarray}
By denoting $p_X=\Tr(\Pi X)$, then the expected risk can be expressed as follows
\begin{align*}
L(f)&=\mathbb{E}_\mu \ell_f(X,Y)\\
&=\mathbb{E}_X \mathbb{E}_{Y|X} \ell_f(X,Y)\\
&=\mathbb{E}_X\left[ p_X \ell_f (X,1) + (1-p_X) \ell_f(X,0)\right]\\
&=:\mathbb{E}_X \ell'_f(X,Y).
\end{align*}
In the third equality we use the fact that the `1' (resp.~`0') outcome occurs with probability $p_X=\Tr(\Pi X)$ (resp.~$1-p_X$). In the last line we introduce the induced loss function $\ell'_f(X,Y):=\left[ p_X \ell_f (X,1) + (1-p_X) \ell_f(X,0)\right]$.
Then for all $X\in\mathcal{X}$, the distance between $\ell'_f$ and $\ell'_g$ can be calculated as
\begin{align*}
|\ell'_f(X,Y)-\ell'_g(X,Y)|&=\left|p_X\left(\ell_f(X,1)-\ell_g(X,1)\right)+(1-p_X)\left(\ell_f(X,0)-\ell_g(X,0)\right)\right|\\
&\leq p_X\left|\ell_f(X,1)-\ell_g(X,1)\right|+(1-p_X)\left|\ell_f(X,0)-\ell_g(X,0)\right|\\
&\leq p_X\cdot \mathsf{L}\left|f(X)-g(X)\right|+(1-p_X)\cdot \mathsf{L}\left|f(X)-g(X)\right|\\
&=\mathsf{L}\left|f(X)-g(X)\right|.
\end{align*}
The second inequality follows from the triangle inequality. The next line is due to the Lipschitz condition.
The above relation shows that the distance $|\ell'_f-\ell'_g|$ can be upper bounded by $\mathsf{L}|f-g|$, which is exactly the same as the upper bound for $|\ell_f-\ell_g|$ (see Eq.~\ref{eq:Lipschitz}). Recall Definition \ref{defn:cn}, it is clearly that the covering numbers with respect to the induced loss function and the original loss function are bounded by the same quantity. Therefore, the generalisation error, Eq.~\eqref{eq:sub_bound2} and the sample complexity do not change in this scenario.

~\\
\textbf{There is noise involved in the measurement procedure.} In this  case, we assume that the training sample is $(X,Y+ \mathbf{n})$, where $Y\equiv\Tr(\Pi X)$ and $\mathbf{n}$ is a random variable that models the measurement error.
Following the same reasoning, we can calculate the {expected risk} as follows
\begin{align*}
L(f)&=\mathbb{E}_\mu \ell_f(X,Y+\mathbf{n})\\
&=\mathbb{E}_X \mathbb{E}_\mathbf{n} \ell_f(X,Y+\mathbf{n})\\
&=:\mathbb{E}_X \ell'_f(X,Y).
\end{align*}
In the last line, we let $\ell'_f(X,Y):=\mathbb{E}_{\mathbf{n}}\ell_f(X,Y+\mathbf{n})$. Thus,
\begin{align*}
|\ell'_f(X,Y)-\ell'_g(X,Y)|&=\left|\mathbb{E}_\mathbf{n}\ell_f(X,Y+\mathbf{n})-\mathbb{E}_\mathbf{n}\ell_g(X,Y+\mathbf{n})\right|\\
&\leq \mathsf{L}\left|\mathbb{E}_\mathbf{n}\left[f(X)-g(X)\right]\right|\\
&=\mathsf{L}\left|f(X)-g(X)\right|.
\end{align*}
Therefore, the original complexity measures (which depends on the distance of the loss function) and the induced sample complexity hold the same.

\section{Learning Quantum Measurements} \label{LQE}

In this section, we follow the quantum learning framework presented in Section \ref{Framework} and explicitly show how to derive the upper bound for the fat-shattering dimension, Rademacher complexity and the covering/entropy number. We then discuss how these complexity measures relate to quantum state discrimination.

Recall that, in the problem of learning an unknown quantum measurement, the goal is to learn a fixed but unknown effect $\Pi\in\mathcal{E}(\mathbb{C}^d)$ through the training data set is $\mathsf{Z}_n=\left\{(\rho_i,\Tr(\Pi\rho_i))\right\}_{i=1}^n$, where $\{\rho_i\}_{i=1}^n\in\mathcal{Q}(\mathbb{C}^d)\equiv\mathcal{X}$ distribute independently according to the unknown measure $\mu$. Note that learning $\Pi$ is equivalent to learning a two-outcome POVM $\{\Pi, \mathcal{I}-\Pi\}$. Due to the correspondence between an quantum effect $E\in\mathcal{E}(\mathbb{C}^d)$ and the linear functional $f_E:\rho\mapsto \langle E,\rho\rangle$ on the input space $\mathcal{X}$ (Proposition~\ref{prop:effect}), we consider the hypothesis set that consists of all quantum effects\footnote{
The hypothesis set can be chosen as a subset of the effects space, to which the target effect $\Pi$ may not belong. Then the goal is to choose an effect in the hypothesis set that approximates the target well. We discuss this issue in Section \ref{algorithms}. Also note that we sometimes denote $\mathcal{F}$ as the subset of $\mathcal{E}(\mathbb{C}^d)$ and sometimes denote it as the linear functionals formed by that subset.}; that is,
\[ \label{eq:hypothesis}
\mathcal{F}=\{f_E:\, E\in \mathcal{E}(\mathbb{C}^d)\}.
\]


In the following, we present our main result to the question: ``\emph{how many quantum states are needed to learn a quantum measurement?}" This is exactly the sample complexity problem introduced in Section \ref{FSLT}. To tackle this problem, we have to estimate the complexity measures that characterise the size of the hypothesis set. 

\subsection{The Fat-Shattering Dimension for Learning Quantum Measurements} \label{SCLQE}

Our first step is to use a common trick in convex analysis; namely, ``symmetrisation" of the state space and the effect space, to embed them into a subset of the Banach space.
In other words, the symmetric convex hull of the state space forms a unit ball of Schatten 1-class:
\[
S_1^d:=\text{conv}(-\mathcal{Q}(\mathbb{C}^d)\cup \mathcal{Q}(\mathbb{C}^d)),
\]
where conv$(\cdot)$ denotes the convex hull operation.
Similarly, we have
\[
S_\infty^d:=\text{conv}(-\mathcal{E}(\mathbb{C}^d)\cup \mathcal{E}(\mathbb{C}^d)).
\]
Now the input space $\mathcal{X}\subset S_1^d$ and the hypothesis set $\mathcal{F}$ consists of linear functionals which can be paremeterised by the elements in $S_\infty^d$. That is,
\[
\mathcal{F}=\{f_E:\, E\in S_\infty^d\}.
\]

The main reason for introducing $S_1^d$  and $S_\infty^d$ is that they are unit balls which are \emph{polar} to each other (through the Hilbert-Schmidt inner product). Thus, we can apply Mendelson and Schechtman's result (Lemma \ref{lemm:Men04}) to estimate the fat-shattering dimension.

The following is our main result in this result.

\begin{theo}[Fat-Shattering Dimension for Learning Quantum Measurements] \label{theo:effect}
For all $0<\epsilon<1/2$, and integer $d\geq 2$, we have
\[
\text{Pdim}(\mathcal{E}(\mathbb{C}^d))\leq d^2,
\]
and
\[
\text{fat}_{\mathcal{E}(\mathbb{C}^d)}(\epsilon,\mathcal{Q}(\mathbb{C}^d))=\min\{O(d/\epsilon^2),d^2\}.
\]
\end{theo}
\begin{proof}
We first present the outline of the proof.
According to the definition of the fat-shattering dimension, it follows that the function fat$_\mathcal{F}(\epsilon)$ is non-increasing with $\epsilon$. Hence, our first objective is to check whether the fat-shattering dimension is unbounded. Equivalently, it suffices to find the pseudo dimension which bounds the fat-shattering dimension (Theorem \ref{theo:AB99}). Secondly, assume there is a set of $n$ points that can be $\epsilon$-shattered; we will find an inequality to relate $n$ with $\epsilon$, which will prove our claim.

(i) Pseudo Dimension: Since $\mathbb{M}_d$ is a vector space with dimension $d^2$ and $S_\infty^d$ is a subset of $\mathbb{M}_d$, we can embed $S_\infty^d$ into a real vector space of dimension $d^2$. Since the function class $\mathcal{F}$ is a subset of a $d^2$-dimensional vector space, by Theorem \ref{theo:pseudo} we obtain Pdim$(\mathcal{F})\leq d^2$.

(ii) Fat-Shattering Dimension:
Consider any set  $\mathcal{S}=\{x_1,\ldots,x_n\}\subset S_1^d$ is $\epsilon$-shattered by $S_\infty^d$, where $n\leq d^2$. 
Denote a \emph{Rademacher series} as $\sum_{i=1}^n \gamma_i x_i$, where $\{\gamma_i\}_{i=1}^n$ are independent and uniform $\{+1,-1\}$ random variables (also called Rademacher random variables). By selecting $a_i=\gamma_i$ in Lemma \ref{lemm:Men04}, we have
\begin{eqnarray} \label{eq:epsilon}
\epsilon n\leq\left\|\sum_{i=1}^n \gamma_i x_i\right\|_1.
\end{eqnarray}
We adopt a probabilistic method to upper bound the right-hand side of Eq.~\eqref{eq:epsilon}. If we can find a quantity $C(n,d)$ that upper bounds $\mathbb{E}\left\|\sum_{i=1}^n \gamma_i x_i\right\|_1$, then there is a realization of $\{\gamma_i\}_{i=1}^n$ such that $\left\|\sum_{i=1}^n \gamma_i x_i\right\|_1\leq C(n,d)$. As a result, it remains to find an upper bound for the expected norm of the Rademacher series $\mathbb{E}\left\|\sum_{i=1}^n \gamma_i x_i\right\|_1$.

In order to upper bound the Rademacher series, we need the powerful \emph{Noncommutative Khintchine Inequalities} \cite{LP91}:

\begin{prop}[Noncommutative Khintchine Inequalities \cite{LP91, HM07}] \label{theo:NKI}
	Let $\{x_i\}_{i=1}^n$ be deterministic $d\times d$ matrices, $\{\gamma_i\}_{i=1}^n$ be independent Rademacher random variables.
	Then
	\[
	\mathbb{E}\left\|\sum_{i=1}^n \gamma_i x_i \right\|_{p}\approx_p
	\left\{\begin{matrix} \left(\|(\sum_{i=1}^n x_i
	x_i^\dagger)^{1/2}\|_{p}^p+\|(\sum_{i=1}^n x_i^\dagger
	x_i)^{1/2}\|_{p}^p\right)^{1/p},\quad\text{if}\quad 2\leq p < \infty \\
	\inf_{x_i=a_i+b_i}\left(\|(\sum_{i=1}^n a_i
	a_i^\dagger)^{1/2}\|_{p}^p+\|(\sum_{i=1}^n b_i^\dagger
	b_i)^{1/2}\|_{p}^p\right)^{1/p},\quad\text{if}\quad1\leq p \leq 2.
	\end{matrix}\right.
	\]
	where $\approx_p$ means that the equality holds up to an absolute constant depending on $p$, and $^\dagger$ denotes the complex conjugate operation.
	
	Note that Haagerup and Musat \cite{HM07} proved that the result also holds as $\{\gamma_i\}_{i=1}^n$ are independent standard complex Gaussian random variables
\end{prop}

Invoking Proposition~\ref{theo:NKI}, we have
\[
\mathbb{E}\left\|\sum_{i=1}^n\gamma_i x_i\right\|_1\lesssim \left\|\left(\sum_{i=1}^n x_i^2\right)^{1/2}\right\|_1.
\]
Since the square operation preserves $S_1^d$, i.e.~$x_i^2\in S_1^d$, for all $x_i\in S_1^d$, by the convexity of $S_1^d$,
we have $\frac1n \sum_{i=1}^n x_i^2\in S_1^d$.
Then the problem is reduced to finding
\[
\max_{\{x_i\}\in S_1^d} \sqrt{n}\left\|\left(\frac1n\sum_{i=1}^n x_i^2\right)^{1/2}\right\|_1=\max_{x\in S_1^d} \sqrt{n}\|\sqrt{x}\|_1,
\]
which is essentially a convex optimisation problem
\[
\max_{x\in S_1^d}\sqrt{n}\sum_{j=1}^d\sqrt{|\lambda_i|}, \text{
subject to } \sum_{j=1}^d |\lambda_j|=1.
\]
Since the square root is concave, we attain the maximum when $|\lambda_j|=1/d$, for $j=1,\ldots,d$.
That is,
\begin{eqnarray} \label{eq:RS1}
\mathbb{E}\left\|\sum_{i=1}^n\gamma_i x_i\right\|_1\lesssim \max_{x\in S_1^d}\sqrt{n}\sum_{i=1}^d\sqrt{\lambda_i}=\sqrt{nd}.
\end{eqnarray}
Consequently, there is a realization of $\{\gamma_i\}_{i=1}^n$ such that $\left\|\sum_{i=1}^n \gamma_i x_i\right\|_1\leq \sqrt{nd},\, \forall x_i\in S_1^d$.
Combined with Eq.~\eqref{eq:epsilon}, we have $n\leq d/\epsilon^2$ which proves our claim.

\end{proof}

In the following proposition, we will demonstrate that the upper bound is tight.

\begin{prop}
Considering a Hilbert space $\mathbb{C}^d$, there exist infinitely
many sets of $d$ quantum states that can be $1/2$-shattered by the
effect space.
\end{prop}
\begin{proof} Consider arbitrary $d$ mutually orthogonal rank-1 projection operators (pure states) $\{\rho_i\}_{i=1}^d$ on $\mathbb{C}^d$ as the input states.
Now for every $B\subseteq \{1,\ldots,d\}$, denote
$f_B:\rho\rightarrow \langle  \sum_{i\in B} \rho_i,\rho \rangle$,
for some $\rho\in \mathcal{Q}(\mathbb{C}^d)$. Note that one can
easily check $\sum_{i\in B}\rho_i\in \mathcal{E}(\mathbb{C}^d)$.
Then for $i\in B$, we have
\begin{align*}
f_B(\rho_i)&=\left\langle  \sum_{i\in B} \rho_i,\rho \right\rangle\\
&=\langle \rho_i,\rho_i\rangle\\
&=1.
\end{align*}
Similarly, $f_B(\rho_i)=0$ if $i\notin B$. As a result,
$\{\rho_i\}_{i=1}^d$ is $1/2$-shattered by $\{f_B\}$.
\end{proof}

\subsection{The Rademacher Complexity} \label{ra_effect}

Following the paradigm in Section \ref{SCLQE}, we calculate the Rademacher complexity of the effect space $\mathcal{E}(\mathbb{C}^d)$ via the duality formula, Theorem \ref{theo:Holder}, and the noncommutative Khintchine inequality, Proposition \ref{theo:NKI}.

\begin{theo}[Rademacher Complexity for Learning Quantum Measurements] \label{theo:ra_effect}
Assume the input space is the state space $\mathcal{X}=\mathcal{Q}(\mathbb{C}^d)$ and the hypothesis set $\mathcal{F}=\{f_E: \forall E\in \mathcal{E}(\mathbb{C}^d)\}$. Then the Rademacher complexity is
	\[
	\mathcal{R}_n(\mathcal{E}(\mathbb{C}^d))=O\left(\sqrt{d}\right).
	\]
\end{theo}
\begin{proof}
	
Recall the definition of the Rademacher complexity (Definition~\ref{def_RCom}). We have
\begin{align*}
\sqrt{n}\mathcal{R}_n(S_\infty^d) &= \mathbb{E}\sup_{E\in S_\infty^d}\left|\sum_{i=1}^n \gamma_i f_E(x_i)\right|\\
&= \mathbb{E}\sup_{E\in S_\infty^d}\left|\sum_{i=1}^n \gamma_i \langle E, x_i \rangle\right|\\
	&=\mathbb{E}\sup_{E\in S_\infty^d}\left|  \left\langle E, \sum_{i=1}^n \gamma_i x_i \right\rangle\right|\\
	&\leq \mathbb{E} \left\|\sum_{i=1}^{n} \gamma_i x_i \right\|_1\\
	&\lesssim \sqrt{nd}.
	\end{align*}
The third line is due to the duality formula (Theorem \ref{theo:Holder}), and the last relation follows from Eq.~\eqref{eq:RS1}. This completes the proof.
\end{proof}

\subsection{The Entropy Number} \label{en_effect}

The covering number (and the related entropy number) follows directly from the Rademacher complexity by the \emph{Sudakov's minoration theorm}.

\begin{theo}[Sudakov's Minoration Theorem]{ \cite{Sud71, VW96, Dud99}} \label{theo:Sud}
Let $\mathcal{T}$ be an index set. Let $X=(X_t)_{t\in\mathcal{T}}$ be a sub-Gaussian process\footnote{A stochastic process is called \emph{sub-Guassian} if there exists $\sigma>0$ such that $\mathbb{E}\exp(\theta X_t)\leq\exp(\sigma^2\theta^2/2)$ for all $\theta\in\mathbb{R}$ and $t\in\mathcal{T}$. Note that both Gaussian process and Rademacher process belong to sub-Gaussian process.} with $L_2$-metric $d_X$ (i.e.~$d_X(s,t)=\|X_s-X_t\|_2$) for $s,t\in\mathcal{T}$). Then for each $\epsilon>0$,
\[
\epsilon(\log \mathcal{N}(\epsilon,\mathcal{T},d_X))^{1/2}\leq C\mathbb{E} \sup_{t\in\mathcal{T}} \|X_t\|_1,
\]
for some constant $C$.
\end{theo}

\begin{coro}[{Entropy Number} for Learning Quantum Measurements] \label{coro:cn_effect}
Assume the input space is the state space $\mathcal{X}=\mathcal{Q}(\mathbb{C}^d)$ and the hypothesis set $\mathcal{F}=\{f_E: \forall E\in \mathcal{E}(\mathbb{C}^d)\}$. Then for each $\epsilon>0$, the covering number of the function class is
	\[
	\log\mathcal{N}_2(\epsilon,\mathcal{E}(\mathbb{C}^d),n)=O(d/\epsilon^2).
	\]
\end{coro}
\begin{proof}
	The upper bound of the empirical $L_2$ entropy number by the Rademacher complexity follows directly from the Sudakov's minoration theorem.
	Denote the (vector-valued) stochastic process by
	\[
	X_f :=\frac1{\sqrt{n}} (\gamma_1 f(x_1),\ldots,\gamma_n f(x_n)),
	\]
	where $x_1,\ldots,x_n$ are independently drawn from $\mathcal{X}$ according to some distribution $\mu$.
	Then the distance measure can be calculated as
	\[
	d_X(f,g)=\left\|X_f-X_g\right\|_2=\frac1{\sqrt{n}}\left(\sum_{i=1}^n \left|f(x_i)-g(x_i)\right|^2\right)^{1/2}=
	\left\| f-g\right\|_{L_2(\mu_n)}.
	\]
	Invoke Theorem \ref{theo:Sud} and \ref{ra_effect} to obtain
	\begin{align*}
	\log\mathcal{N}(\epsilon,\mathcal{F},L_2(\mu_n)) &= \log\mathcal{N}(\epsilon,\mathcal{F},d_X)\\
	&\leq C^2\frac{\left(\mathbb{E}\sup_{f\in\mathcal{F}}\|X_f\|_1\right)^2}{\epsilon^2}\\
	&=C^2\frac{\mathcal{R}_n(\mathcal{F})^2}{\epsilon^2}\\
	&\leq C^2\frac{d}{\epsilon^2}.
	\end{align*}
	Note that the right-hand side in the last line does not depend on the distribution $\mu$. Hence the entropy number $\log\mathcal{N}_2(\epsilon,\mathcal{F},n)=\sup_{\mu_n} \log\mathcal{N}(\epsilon,\mathcal{F},L_2(\mu_n))=O(d/\epsilon^2)$ follows.
\end{proof}

The pseudo dimension of the effect space Pdim$(\mathbb{C}^d)=d^2$ means that we need $d^2$ parameters to exactly determine a POVM element.
Note that it coincides the number of measurements in the quantum measurement tomography (since $\mathcal{E}(\mathbb{C}^d)$ lies in a $d^2$-dimensional real vector space).
However, if we relax the criterion by tolerating an $\epsilon$ accuracy, then the effect space can be covered by $\mathcal{N}_2(\epsilon,\mathcal{E}(\mathbb{C}^d))=\exp(d/\epsilon^2)$ balls each with radius $\epsilon$.
In other words, we need $\log \mathcal{N}_2(\epsilon,\mathcal{E}(\mathbb{C}^d))=d/\epsilon^2$ samples to identify which ball the target POVM element lies in. That is the meaning of the entropy number.
By applying the quantum learning model to quantum measurement tomography, we can specify a ``PAC" candidate POVM element with accuracy $\epsilon$ and confidence $\delta$ with only $d/\epsilon^2$ samples, which quadratically speed-up the original scheme.

\subsection{The Relationship to Quantum State Discrimination} \label{RQSD}
\emph{Quantum State Discrimination} studies how to optimally distinguish a set of quantum states according to a figure of merit \cite{CH13, CDH14}. 

There are nevertheless some limitations in quantum state discrimination because the states cannot always be perfectly discriminated. Moreover, it may not be necessary to find the exact state in some scenario. Therefore, Zhang and Ying \cite{ZY02} considered \emph{quantum set discrimination}, where the goal is to identify which set the given state belongs to. 
Now we relate the concepts of the fat-shattering dimension to \emph{quantum set discrimination}.

\begin{defn}[$\epsilon$-separable Set]
	A set $\mathcal{S}=\{x_1,\ldots,x_n\}\subset\mathbb{M}_d$ is \emph{$\epsilon$-(linearly) separable} with respect to the set $\mathcal{W}\subseteq\mathbb{M}_d$ if and only if for any subset $B\subseteq \mathcal{S}$ there exists an $\epsilon$-strip which separates $B$ from its complement $\mathcal{S}\setminus B$. In other words, there exist $w\in\mathcal{W}$ and $a\in\mathbb{R}$ such that  $\langle w, x\rangle\geq a+\epsilon/2$ when $x\in B$ and $\langle w, x\rangle\leq a-\epsilon/2$ when $x\in \mathcal{S}\setminus B$.
\end{defn}

It is not difficult to see that an $2\epsilon$-separable set correspond to the task of quantum set discrimination with ensemble $\mathcal{S}=\{x_1,\ldots,x_n\}$, where the error probability that a given state can be classified to a set is no greater than $(1-\epsilon)/2$.  
One interesting question to ask is what the maximum cardinality of the $2\epsilon$-separable set is. The following proposition shows that the fat-shattering dimension equals this quantity.

\begin{prop}
	Denote the function class $\mathcal{F}=\{\rho\rightarrow \langle E, \rho\rangle : E\in\mathcal{E}(\mathbb{C}^d)\}$. 
	Assume there exists a set $\mathcal{S}=\{x_1,\ldots,x_n\}\subset \mathcal{Q}(\mathbb{C}^d)$ that is $2\epsilon$-separable with respect to $\mathcal{E}(\mathbb{C}^d)$.
	Then the maximum cardinality of the set $\mathcal{S}$ is fat$_\mathcal{F}(\epsilon)$.
\end{prop}

\begin{proof} Recall from Definition \ref{defn:fat} that the set $\mathcal{S}=\{x_1,\ldots,x_n\}$ is $2\epsilon$-separable with respect to $\mathcal{E}(\mathbb{C}^d)$ if and only if $\underline{\text{fat}}_\epsilon(\mathcal{F})\geq n$. Then the proposition is equivalent to show that fat$_\epsilon(\mathcal{F})=\underline{\text{fat}}_\epsilon(\mathcal{F})$.
	
	Because $\underline{\text{fat}}_\epsilon(\mathcal{F})\leq \text{fat}_\epsilon(\mathcal{F})$ by definition, it suffices to show $\underline{\text{fat}}_\epsilon(\mathcal{F})\geq \text{fat}_\epsilon(\mathcal{F})$.
	Given $\epsilon>0$, choose a set $\mathcal{S}=\{x_1,\ldots,x_n\}$ with the largest integer $n$ such that $\mathcal{S}$ is $\epsilon$-shattered by $\mathcal{F}$ (with $\{s_i\}_{i=1}^n$ witnessing the shattering).
	Without loss of generality, we assume some $s_i\neq 1/2$. We then choose an arbitrary subset $B\subseteq \{1,\ldots, n\}$ that contains $i$. By the definition of fat-shattering dimension, there exists $s_i:=s(x_i)$ such that there is some function $E_B\in \mathcal{F}$ for each set $B\subset \mathcal{S}$ so that $\langle E_B,x_i\rangle \geq s_i+\epsilon$, if $i\in B$.
	Also, we have $\langle E_{\bar{B}},x_i\rangle \leq s_i-\epsilon$, where $\bar{B}=\mathcal{S}\setminus B$.
	Now denote $\overline{E_{\bar{B}}}:=\mathcal{I}-E_{\bar{B}}$ such that
	\[
	\langle \overline{E_{\bar{B}}},x_i\rangle=1-\langle E_{\bar{B}},x_i\rangle\geq 1-s_i+\epsilon.
	\]
	Since $\mathcal{F}$ is convex, set
	$E_{B}':=\frac12(E_B+\overline{E_{\bar{B}}})\in\mathcal{F}$ which
	satisfies
	\[
	\langle E_{B}', x_i\rangle\geq 1/2 +\epsilon.
	\]
	Similarly, let $E_{\bar{B}}':=\mathcal{I}-E_{B}'$, we have
	\[
	\langle E_{\bar{B}}', x_i\rangle\leq 1/2 -\epsilon.
	\]
	The same argument holds for other $s_i\neq 1/2$.
	It follows that the level fat-shattering dimension (witnessed by $1/2$) also achieves the cardinality $n$ of the $\epsilon$-shattered set, which completes the proof.
\end{proof}

\section{Learning Quantum States} \label{LQS}

In this section, we consider the problem of learning an unknown quantum state $\rho'\in\mathcal{Q}(\mathbb{C}^d)$ through the training data set $\mathsf{Z}_n=\left\{\left(E_i,\Tr(\rho' E_i)\right)\right\}_{i=1}^n$, where $\{E_i\}_{i=1}^n\in\mathcal{X}=\mathcal{E}(\mathbb{C}^d)$ are independently sampled according to an unknown distribution $\mu'$. 
By Proposition~\ref{prop:state}, the hypothesis set consists of the linear functional $f_\rho:E\mapsto \langle E,\rho\rangle$ on $\mathcal{E}(\mathbb{C}^d)$: 
\[
\mathcal{F}'=\{f_\rho : \forall \rho\in \mathcal{Q}(\mathbb{C}^d)\}.
\]


Similarly, we embed the input space into the unit ball of Schatten $\infty$-class, i.e.~$\mathcal{X}=S_\infty^d$. Then the hypothesis set is the collection of  linear functionals on the input space, i.e.~$S_1^d$.
In the following, we aim to calculate the complexity measures of $S_1^d$, which characterise the sample complexity of learning quantum states. It is interesting to see that the proofs derived in this section (i.e.~the complexity measures of learning quantum states) parallel with that in the previous section (i.e.~the complexity measures of learning quantum measurements) due to the duality relation in Theorem~\ref{theo:duality}. Finally, we discuss the relationship of the fat-shattering dimension with quantum random access codes.


\subsection{The Fat-Shattering Dimension for Learning Quantum States} \label{SCLQS}
Under the framework presented in Section \ref{Framework}, we characterising the input space $\mathcal{X} \subset S_\infty^d$ and the hypothesis set $\mathcal{F}'$ consisting of the linear functionals with elements in $S_1^d$. That is,
\[
\mathcal{F}'=\{ f_\rho : \rho\in S_1^d\}.
\]

Therefore, we have the main result of deriving the fat-shattering dimension of the state space.
\begin{theo}[Fat-Shattering Dimension for Learning Quantum States] \label{theo:state}
For all $0<\epsilon< 1/2$ and integer $d\geq 2$, we have
\[
\text{Pdim}(\mathcal{Q}(\mathbb{C}^d))\leq d^2-1,
\]
and
\[
\text{fat}_{\mathcal{Q}(\mathbb{C}^d)}(\epsilon,\mathcal{E}(\mathbb{C}^d))=\min\{O(\log d/\epsilon^2),d^2-1\}.
\]
\end{theo}
\begin{proof}
Following the same fashion as in the proof of Theorem \ref{theo:effect}, we first estimate the pseudo dimension and then the fat-shattering dimension.

(i) Pseudo Dimension:
The state space lies in the set $\{x\in\mathcal{M}^d:\, \|x\|_1=1\}$, which is the sphere of $S_1^d$, i.e.~$\mathcal{Q}(\mathbb{C}^d)\subset \partial S_1^d$. Since $\partial S_1^d$ can be embedded into a real vector space of dimension $d^2-1$, we have Pdim$(\mathcal{Q}(\mathbb{C}^d))\leq d^2-1$.

(ii) Fat-Shattering Dimension:
For every $\{x_i\}_{i=1}^n\in S_\infty^d$, we have to calculate the Rademacher series $\mathbb{E}\left\|\sum_{i=1}^n \gamma_i x_i\right\|_\infty$.
However, in the scenario of learning quantum states the input space lies in the Schatten $\infty$-class.
We have to estimate the spectral norm of the Rademacher series.
Benefiting from the recent development of matrix concentration inequalities, Tropp \cite{Tro11d} proved the following results:
\begin{prop}[Upper Bound for Rademacher Series \cite{Tro11d}] \label{theo:RS}
Consider a finite sequence $\{x_i\}$ of deterministic Hermitian matrices with dimension $d$, and let $\{\gamma_i\}$ be independent Rademacher variables. Form the matrix Rademacher series
\[
Y=\sum_i \gamma_i x_i.
\]
Compute the variance parameter
\[
\sigma^2=\sigma^2(Y)=\|\mathbb{E}\left(Y^2\right)\|_\infty.
\]
Then
\[
\mathbb{E}\|Y\|_\infty\leq\sqrt{2\sigma^2 \log d}.
\]
Note that the result also holds for the case $\{\gamma_i\}$ being standard complex Gaussian variables.
\end{prop}

Invoking Tropp's development of matrix concentration inequalities (see Proposition \ref{theo:RS}), we have
\begin{eqnarray} \label{eq:bound_inf}
\mathbb{E}\left\|\sum_{i=1}^n \gamma_i x_i\right\|_\infty \leq \sqrt{2\sigma^2 \log d},
\end{eqnarray}
where $\sigma^2:=\left\|\mathbb{E}\left(\sum_{i=1}^n \gamma_i x_i\right)^2\right\|_\infty$.
Straightforward computation shows that
\[
\sigma^2=\left\|\mathbb{E}\left(\sum_{i=1}^n \gamma_i x_i\right)^2\right\|_\infty=\left\|\mathbb{E} \left( \sum_{i,j} \gamma_i \gamma_j x_i x_j\right)\right\|_\infty=\left\|\sum_{i=1}^n x_i^2\right\|_\infty \leq n.
\]
We get
\[
\mathbb{E}\left\|\sum_{i=1}^n \gamma_i x_i\right\|_\infty \leq \sqrt{2n \log
d}.
\]
Then there is a realization of $\{\gamma_i\}_{i=1}^n$ such that $\left\|\sum_{i=1}^n \gamma_i x_i\right\|_\infty\leq\sqrt{2n\log d},\, \forall x_i\in S_\infty^d$.

From Lemma \ref{lemm:Men04}, by selecting $a_i=\gamma_i$, $\epsilon n\leq\left\|\sum_{i=1}^n \gamma_i x_i\right\|_\infty$.
Combining the inequalities, we have $n\leq O(\log d/\epsilon^2)$ completing the proof.
\end{proof}

\subsection{The Rademacher Complexity} \label{ra_state}

By repeating the procedure introduced in Section \ref{ra_effect}, we can compute the Rademacher complexity of the state space.
\begin{theo}[Rademacher Complexity for Learning Quantum States] \label{theo:ra_state}
	Assume the input space is the effect space $\mathcal{X}=\mathcal{E}(\mathbb{C}^d)$. The hypothesis set $\mathcal{F}$ defined on $\mathcal{X}$ is the state space $\mathcal{Q}(\mathbb{C}^d)$. Then the Rademacher complexity of hypothesis set is
	\[
	\mathcal{R}_n(\mathcal{Q}(\mathbb{C}^d))=O\left(\sqrt{\log d}\right).
	\]
\end{theo}
\begin{proof} 
	
	Recall from the definition of the Rademacher complexity. We have
	\begin{align*}
	\sqrt{n}\mathcal{R}_n(S_1^d) &= \mathbb{E}\sup_{\rho\in S_1^d}\left|\sum_{i=1}^n \gamma_i f_\rho(E_i)\right|\\
	&= \mathbb{E}\sup_{\rho\in S_1^d}\left|\sum_{i=1}^n \gamma_i \left\langle E_i, \rho \right\rangle\right|\\
	&=\mathbb{E}\sup_{\rho\in S_1^d}\left|  \left\langle \sum_{i=1}^n \gamma_i E_i,\rho \right\rangle\right|\\
	&\leq \mathbb{E} \left\|\sum_{i=1}^{n} \gamma_i E_i \right\|_\infty\\
	&\lesssim \sqrt{n\log d}.
	\end{align*}
	The forth line is due to the duality formula, Theorem \ref{theo:Holder}. The last relation follows from Eq.~\eqref{eq:bound_inf}, which completes the proof.
\end{proof}

\subsection{The Entropy Number} \label{cn_state}

\begin{coro}
	[Entropy Number for Learning Quantum States] \label{coro:cn_state}
	Assume the input space is $\mathcal{X}=\mathcal{E}(\mathbb{C}^d)$. The function class $\mathcal{F}$ defined on $\mathcal{X}$ is the state space $\mathcal{Q}(\mathbb{C}^d)$. Then for each $\epsilon>0$, the covering number of the function class is
	\[
	\log\mathcal{N}_2(\epsilon,\mathcal{Q}(\mathbb{C}^d),n)=O(\log d/\epsilon^2).
	\]
\end{coro}

Compared with the entropy number of the effect space, the result of the state space is proportional to the logarithmic dimension.
The intuition behind this is that the unit ball of Schatten $\infty$-class is much larger than the unit ball of Schatten $1$-class.
Thus, it requires more $\epsilon$-radius ball to cover the whole effect space than the state space.
From the volumetric perspective, the fact will be more evident.
Denote $|\cdot|$ as the Lebesgue measure on the Banach space of the Schatten class. It can be calculated that
	\[
	\frac{|\mathcal{E}(\mathbb{C}^d)|^{1/d^2}}{|\mathcal{Q}(\mathbb{C}^d)|^{1/(d^2-1)}}\simeq \left(\frac{|S_\infty^d|}{|S_1^d|}\right)^{1/d^2}\simeq d,
	\]
which shows that the volume of the effect space is essentially exponential (in the dimension $d$) to the state space.
Recall that the complexity measures are the quantity to estimate the effective size of the hypothesis set.
Accordingly, it is reasonable that the complexity measures of the effect space are exponentially compared with that of the state space.
In other words, the results of Theorem \ref{theo:effect} demonstrate the richness of the effect space.

\subsection{The Relationship to Quantum Random Access Codes} \label{RQRAC}
The learnability of quantum states was first addressed by Aaronson
\cite{Aar07}. Ingeniously, he applied the results of Quantum Random
Access Coding \cite{ANT+02} to provide an information-theoretic
upper bound on the fat-shattering dimension for learning $m$-qubit
quantum states. We first give the definitions of QRA codes then discuss
Aaronson's result.

\begin{defn}[Quantum Random Access Coding] \label{defn:QRA}
	An $(n,m,p)$-QRA coding is a function that maps $n$-bit strings $x\in\{0,1\}^n$ to $m$-qubit states $\rho_x$ satisfying the following: For every $i\in\{1,\ldots,n\}$ there exists a POVM $E^i=\{E_0^i,E_1^i\}$ such that Tr$(E_{x_i}^i\rho_x)\geq p$ for all $x\in\{0,1\}^n$, where $x_i$ is the $i$-th bit of $x$.
\end{defn}

If there exists an $(n,m,p)$-QRA coding, we have the fact that the sets $\{E_i\}_{i=1}^n$ are $(p-1/2)$-shattered by $\{\rho_y\}$ and the constant value $1/2$ witnesses the shattering. That is,
\begin{eqnarray} \label{eq:aar}
m\geq (1-H(\epsilon+1/2))n\geq c\cdot\epsilon^2n.
\end{eqnarray}
Therefore, the inequality gives an upper bound on the level
fat-shattering dimension, i.e.~
$\underline{\text{fat}}_{\mathcal{Q}(\mathbb{C}^d)}(p-1/2)=O(m/\epsilon^2)$.
Conversely, fat-shattering dimension with scale $(p-1/2)$ does not
guarantee the existence of an $(n,m,p)$-QRA coding (since there may
be some $s_i<1/2$), while provide an upper bound on the success
probability $p$ if it exists.

However, in the case that functions in $\mathcal{F}$ have a bounded range of $[0,1]$, Gurvits \cite{Gur97} utilised the Pigeonhole principle to relate the level fat-shattering dimension with the fat-shattering dimension.

\begin{theo}[Gurvits \cite{Gur97}]
	For any hypothesis set $\mathcal{F}$ consisting of $[0,1]$-valued functions, we have
	\begin{eqnarray} \label{eq:level}
	(2(1-2\epsilon)/\epsilon)^{-1}\text{fat}_{\mathcal{F}}(2\epsilon)\leq \underline{\text{fat}}_\mathcal{F}(\epsilon/2) \leq \text{fat}_{\mathcal{F}}(\epsilon/2).
	\end{eqnarray}
\end{theo}

By definition, $\underline{\text{fat}}_\mathcal{F}(\epsilon)\leq$ fat$_\mathcal{F}(\epsilon)$. However, from the above theorem, the dependencies on the dimension $d$ are of the same order for both the level fat-shattering dimension and the fat-shattering dimension.
Consequently, from Eq.~\eqref{eq:aar} we have $\underline{\text{fat}}_\mathcal{F}(\epsilon)=O(m/\epsilon^2)$, which leads to fat$_\mathcal{F}(\epsilon)=O(m/\epsilon^2)$ according to the inequalities in Eq.~\eqref{eq:level}. Thus we recover Aaronson's result.

\begin{theo}[Aaronson  \cite{Aar07}]{} \label{aar07}
The fat-shattering dimension for learning the class of all $m$-qubits, $\mathcal{F}$, is fat$_\mathcal{F}(\epsilon)=O(m/\epsilon^2)$.
\end{theo}

We remark that it is unknown whether $\underline{\text{fat}}_\mathcal{F}(\epsilon)=$ fat$_\mathcal{F}(\epsilon)$ for $\mathcal{F}=\mathcal{Q}(\mathbb{C}^d)$.

\begin{prop}
	There is no $(2^{2m}, m,p)$-QRA coding for $1/2<p\leq 1$ and positive integer $m$.
\end{prop}

{Hayashi \emph{et al.} \cite{HIN+06} showed that there is no $(2^{2m},m,p)$-QRA coding for $1/2<p\leq 1$.} This result can be directly derived from Theorem \ref{theo:state}, which shows that Pdim$(\mathcal{Q}(\mathbb{C}^d))\leq d^2-1$. The dimension $d$ of $m$-qubit is $2^m$. Then the upper bound of the pseudo dimension shows that there is no $d^2=2^{2m}$ two-outcome POVMs that can be shattered (by the function class of the state space), which coincides with Hayashi \emph{et 	al's} result.

\section{The Algorithms for Quantum Machine Learning} \label{algorithms}

In the previous sections, we demonstrate the information-theoretical analysis of the quantum learning problems. In this section, provide a constructive way to implement quantum ML tasks by representing the learning framework in Bloch space. 

We gather all the materials and derivations concerning the Bloch-sphere representation into Appendix \ref{BR}.
Recall from Eq.~\eqref{eq:rank_k_bloch} that the function class of rank-$k$ effects and their mixture can be represented as the following affine functional:
\[
\mathcal{F}_k=\text{conv}\left(\{\mathbf{r}\mapsto\frac{k}{d}\left(1+(d-1)\mathbf{r}\cdot\mathbf{n}_{(k)}\right)\}\right),
\]
where $\mathbf{r}$ is the Bloch vector of the quantum state; $\mathbf{n}_{(k)}$ (see Eq.~\eqref{eq:nk}) parameterises the function in the hypothesis set $\mathcal{F}_k$.
Moreover, it can in turn be written as
\[
\mathcal{F}_k=\sigma(\mathbf{v}\cdot\mathbf{r}+v_0),
\]
where $\sigma:\mathbb{R}\rightarrow\mathbb{R}$ is called the \emph{activation function}. The Bloch vector $\mathbf{r}\in\mathbb{R}^{d^2-1}$ is the input vector; $[v_0, \mathbf{v}]\in\mathbb{R}^{d^2}$ is the input weights. Each map $\mathbf{r}\mapsto \sigma(\mathbf{v}\cdot\mathbf{r}+v_0)$ can be thought of as a function computed by the linear perceptron. Using the terminology from the theory of neural network \cite{AB99}, each $\mathcal{F}_k$ is called the \emph{single-layer neural network} (see Appendix \ref{nn} for more details).

Considering the function class of the whole effect space, we exploit the convexity of the effect space, and obtain the following result:
\[
\mathcal{F}=\sum_{k=0}^d \frac{w_k\cdot
k}{d}\left(1+(d-1)\mathbf{r}\cdot\mathbf{n}_{(k)}\right)=:\frac{1}{d}\left(n_0+(d-1)\mathbf{r}\cdot\mathbf{n}\right),
\]
where $\sum_{k=0}^d w_k=1$. This is called the \emph{two-layer neural network} (also called the \emph{single-hidden layer net}). Based on this formulation, the tasks of learning quantum
measurements can be implemented by existing neural network algorithms or other classical ML algorithms. We note that the neural network formulation for learning quantum states follows in the same way by virtue of the duality.

Additionally, the fat-shattering dimension for $\mathcal{F}_k$ can easily be bounded from the classical results in neural networks. We have the following corollary.

\begin{coro} \label{coro:k}
Suppose the hypothesis set $\mathcal{F}_k$ consists of rank-$k$ projection operators and their mixture. We have
\[
\text{fat}_{\mathcal{F}_k}(\epsilon)\leq \frac{k(d-1)(d-k)}{{(d\epsilon)}^2},\, k=\{0,1,\ldots,d\}.
\]
\end{coro}
\begin{proof} Since $\mathcal{F}_k$ is a linear function class on $\mathbb{R}^{d-1}$, invoking the classical results from Anthony and Bartlett \cite{AB99}:
\[
\text{fat}_\mathcal{F}(\epsilon)\leq \frac{a^2b^2}{\epsilon^2},
\]
where $\mathcal{F}=\{\mathbf{w}\mapsto\langle \mathbf{w},\mathbf{x}\rangle:\, \|\mathbf{x}\|_2\leq b,\|\mathbf{w}\|_2\leq a,\, \mathbf{x},\mathbf{w}\in\mathbb{R}^{d^2-1}\}$.

Therefore, it remains to calculate the coefficients in Eq.~\eqref{nk}. Since $\|\mathbf{r}\|_2\leq 1$, and
\[
\left\|\frac{k(d-1)}d \sqrt{\frac{d-k}{k(d-1)}}\right\|_2=\sqrt{\frac{k(d-1)(d-k)}{d^2}},
\]
the result follows.
\end{proof}

We can see from the corollary that the fat-shattering dimension increases when the the rank $k$ approaches a half of the Hilbert space dimension $d$, which means that the classes $\{\mathcal{F}_k\}$ form a hierarchical structure. Operationally, the hypothesis set $\mathcal{F}_1$ can be chosen at first. It can then be enlarged into
conv$(\mathcal{F}_0\cup\mathcal{F}_1\cup\mathcal{F}_2)$ and so forth until the whole effect space is considered. This is called the structural risk minimisation (SRM \cite{Vap98}), and is usually adopted in classical ML to avoid overfitting.
Here we give two examples to illustrate the concepts in Corollary \ref{coro:k}.

\begin{exam}[Learning rank-$1$ Projection Valued Measures (PVMs): Qubit system attains the upper bound] \label{ex:qubit}
The fat-shattering dimension of rank-$1$ projection operators and their mixture in a qubit system can be bounded by
\[
\text{fat}_{\mathcal{F}_{1}}(\epsilon)\leq \frac{(N-1)^2}{(N\epsilon)^2}=\frac{1}{4\epsilon^2}.
\]
Consider two quantum states $\rho_{\mathbf{r}_1}=|1\rangle\langle1|$, $\rho_{\mathbf{r}_2}=|-\rangle\langle-|$ with corresponding Bloch vectors $\mathbf{r}_1=(0,0,-1)$, $\mathbf{r}_2=(-1,0,0)$.
To shatter these two quantum states, we construct four quantum effects with the Bloch vectors:
\begin{align*}
\mathbf{n}_{00}=\frac{1}{\sqrt{2}}(1,0,1),\, \mathbf{n}_{10}=\frac{1}{\sqrt{2}}(1,0,-1),\\
\mathbf{n}_{11}=\frac{1}{\sqrt{2}}(-1,0,-1),\;
 \mathbf{n}_{01}=\frac{1}{\sqrt{2}}(-1,0,1).
\end{align*}
Since the angles between the states and effects are either $\pi/4$ or $3\pi/4$, we have
\begin{align*}
\left(\Tr(E_{\mathbf{n}_{00}}\rho_{{\mathbf{r}}_1}),\Tr(E_{\mathbf{n}_{00}}\rho_{{\mathbf{r}}_2})\right)=(\frac12(1-\frac{1}{\sqrt{2}}),\frac12(1-\frac{1}{\sqrt{2}})),\, \left(\Tr(E_{\mathbf{n}_{10}}\rho_{{\mathbf{r}}_1}),\Tr(E_{\mathbf{n}_{10}}\rho_{{\mathbf{r}}_2})\right)=(\frac12(1+\frac{1}{\sqrt{2}}),\frac12(1-\frac{1}{\sqrt{2}})),\\
\left(\Tr(E_{\mathbf{n}_{11}}\rho_{{\mathbf{r}}_1}),\Tr(E_{\mathbf{n}_{11}}\rho_{{\mathbf{r}}_2})\right)=(\frac12(1+\frac{1}{\sqrt{2}}),\frac12(1+\frac{1}{\sqrt{2}})),\, \left(\Tr(E_{\mathbf{n}_{01}}\rho_{{\mathbf{r}}_1}),\Tr(E_{\mathbf{n}_{01}}\rho_{{\mathbf{r}}_2})\right)=(\frac12(1-\frac{1}{\sqrt{2}}),\frac12(1+\frac{1}{\sqrt{2}})).
\end{align*}
Clearly these four quantum effects $\frac{1}{2\sqrt{2}}$-shatter $(\mathbf{r}_1,\mathbf{r}_2)$ and achieve the fat-shattering dimension fat$_{\mathcal{F}_1}(\frac{1}{2\sqrt{2}})=2$.

The case of three quantum states follows similarly. Consider $\mathbf{r}_1=(1,0,0)$, $\mathbf{r}_2=(0,1,0)$, $\mathbf{r}_3=(0,0,1)$, and $\mathbf{n}_{ijk}=(i,j,k)$ for $i,j,k\in\{0,1\}$.
With some calculations, the eight quantum effects $\frac{1}{2\sqrt{3}}$-shatter $(\mathbf{r}_1,\mathbf{r}_2,\mathbf{r}_3)$ and achieve the fat-shattering dimension fat$_{\mathcal{F}_{1}}(\frac{1}{2\sqrt{3}})=3$.

It is worth emphasising that the dual problem of learning quantum
states is equivalent to learning quantum measurements when the hypothesis
set consists of rank-$1$ projections and their mixture. The reason is
that the two mathematical objects are exactly the same, i.e.~
conv($\mathcal{F}_1)=\mathcal{Q}(\mathbb{C}^d)$. In this scenario,
the dual problem has the same results, which is optimal in the sense
of Quantum Random Access codes (i.e.~(2,1,0.85)-QRA codes \cite{ALM+08}).
Furthermore, we note that the measurements in the (2,1,0.85)-QRA codes 
and the input states $(\rho_{\mathbf{r}_1},\rho_{\mathbf{r}_1}^\perp),\,(\rho_{\mathbf{r}_2},\rho_{\mathbf{r}_2}^\perp)$ in this example are mutually unbiased bases
(MUB) which attain the upper bound of the qubit system.
\end{exam}

\begin{exam}[Rank equals a half the Hilbert space dimension]
Consider a quaternary Hilbert space, i.e.~$\mathbb{C}^4$.
First, we show that there exist no two quantum states that can be $1/2$-shattered by the convex hull of rank-$1$ projection operators.
Consider two arbitrary different quantum states $\mathcal{S}=\{\rho_i\}_{i=1}^2$.
If the function class $\mathcal{F}_1$ can $1/2$-shatter the set $\mathcal{S}$, then there must be an effect $E\in\mathcal{F}_{1}$ such that Tr$(E\rho_1)=\Tr(E\rho_2)=1$.
Clearly, it can be achieved only when $E$ is a rank-$1$ projection and the two quantum states are both equal to $E$, which contradicts the assumption.

Second, we show there exist two quantum states that can be $1/2$-shattered by the rank-$2$ projection operators.
Assume $\rho_i=|i-1\rangle\langle i-1|$, $i=1,2$. We construct four quantum effects as follows:
\begin{align*}
E_{11}=\begin{pmatrix} 1 &  &  & \\  & 1 &  & \\  &  & 0 & \\  &  &  & 0\end{pmatrix}, \,E_{01}=\begin{pmatrix} 0 &  &  & \\  & 1 &  & \\  &  & 1 & \\  &  &  & 0\end{pmatrix}, \, E_{10}=\begin{pmatrix} 1 &  &  & \\  & 0 &  & \\  &  & 1 & \\  &  &  & 0\end{pmatrix}, \,
E_{00}=\begin{pmatrix} 0 &  &  & \\  & 0 &  & \\  &  & 1 & \\  &  &  & 1\end{pmatrix}
\end{align*}
in the computational basis.
The two quantum states can then be $1/2$-shattered by these four quantum effects.
This example demonstrates that the set of rank-$2$ projections is richer than the set of rank-$1$ projections in terms of the complexity measures.
\end{exam}

\begin{remark}
The readers may contemplate the pros and cons of Bloch-sphere representation when analysing the fat-shattering dimension. Indeed, Bloch-sphere representation provides a geometric picture so that we have more concrete ideas of the linear relation between quantum measurements and states. Furthermore, in Example 
\ref{ex:qubit} we see how the extreme points (projection operators) and MUB play the role in the fat-shattering dimension. However, it is difficult to fully characterise the region of the Bloch space. To the best of our knowledge, the most convenient metric used in Bloch-sphere representation is the Euclidean norm, which corresponds to the Hilbert-Schmidt norm (Schatten 2-norm) in the state space, i.e.~
\[
\|\rho_{\mathbf{r}_1}-\rho_{\mathbf{r}_2}\|_{\text{HS}}=\sqrt{\frac{d-1}{2d}}\|\mathbf{r}_1-\mathbf{r}_2\|_2.
\]
Recalling that conv$\left(-\mathcal{Q}(\mathbb{C}^d)\cup\mathcal{Q}(\mathbb{C}^d)\right)=S_1^d\subset S_2^d \subset S_\infty^d=\text{conv}\left(-\mathcal{E}(\mathbb{C}^d)\cup\mathcal{E}(\mathbb{C}^d)\right)$, the Hilbert-Schmidt norm is not efficient in characterising the state space (that is why some regions in the Bloch sphere are not representative as valid states). On the other hand, the unit ball of Schatten 2-class is not sufficient to contain $S_\infty^d$, so we have to scale up the Hilbert-Schmidt norm by a factor $\sqrt{d}$ (since $\|\cdot\|_2\leq \sqrt{d}\|\cdot\|_\infty$). 
Then we may overestimate the effective size of the effect space.
As a result, directly analysing the linear functionals between $S_1^d$ and $S_\infty^d$ is the most efficient way of calculating the fat-shattering dimension.
We emphasise that with Bloch-sphere representation, all the quantum measurements/states are transformed into Euclidean space, where existing ML algorithms (e.g.~perceptron learning algorithm, neural network, SVM, etc.) can be applied to conduct the learning tasks. It is also worth considering other metrics (e.g.~Bures metric, or other $\ell_p$ norms in Bloch-sphere representation) and parameterisation methods (e.g.~Weyl operator basis, polarisation operator basis, Majorana representation, etc.) in our quantum ML framework. We leave it as future work.
\end{remark}

When learning an $(M+1)$-outcome POVM measurement
$\{\Pi_j\}_{j=0}^M$, with $\sum_{j=0}^M \Pi_j=\mathcal{I}$, we can
simply follow the procedure discussed so far. Now the training data
set consists of $\{(\rho_i,
\Tr(\mathbf{\Pi}\rho_i)\}_{i=1}^n$, where
\[
\Tr(\mathbf{\Pi}\rho_i):=\left(\Tr(\Pi_1\rho_i),\ldots,\Tr(\Pi_n\rho_i)\right).
\]
This is called \emph{multi-target prediction} or \emph{multi-label
classification}. Each target $\Pi_j$ can be independently learned by
the individual function class $\mathcal{F}$.

It is worth mentioning that Gross and Flammia \emph{et al.} \cite{GLF+10, FGL+12} proposed a quantum state tomography method via \emph{compressed sensing}, which is similar to our setting of learning quantum states.
The main goal of the work is to concentrate on states $\rho$ that can be well approximated by density matrices of rank $r\ll d$ and to reconstruct a density matrix  $\widehat{\rho}$ based on $m$ randomly sampled Pauli operators.
With certain constraint coefficients $\lambda$ and $m\geq Crd\log^6 d$, they show
\[
\|\widehat{\rho}-\rho\|_1\leq C_0r\lambda+ C_1\|\rho_c\|_1,
\]
where $\rho_c=\rho-\rho_r$ is the residual part and $\rho_r$ is the best rank-$r$ approximation to $\rho$.

\section{Conclusions} \label{conclusion}

\begin{table}[h]\center
\caption{The Complexity Measures of The Quantum Learning Problems.
} \label{table:QML}
\begin{tabular}{ccc}
\toprule
   & Learning Quantum Measurements & Learning Quantum States\\
\midrule
Pseudo Dimension & $d^2$ & $d^2-1$\\
Fat-Shattering Dimension fat$_\mathcal{F}(\epsilon)$ & $d/\epsilon^2$ & $\log d/\epsilon^2$     \\
Uniform Entropy Number $\log\mathcal{N}_2(\epsilon,\mathcal{F})$ & $d/\epsilon^2$ & $\log d/\epsilon^2$\\
Rademacher Complexity $\mathcal{R}_n(\mathcal{F})$ & $\sqrt{d}$ & $\sqrt{\log d}$\\
Sample Complexity $m_\mathcal{F}(\epsilon,\delta)$ & $\max\{d,\log(1/\delta)\}/\epsilon^2$ & $\max\{\log d,\log(1/\delta)\}/\epsilon^2$\\
\bottomrule
\end{tabular}
\end{table}

In this paper, we formalise the problems of learning quantum measurements and quantum states and  analyse the learnability. We solved the sample complexity problems for learning quantum measurements and quantum states. In the scenario of learning (two-outcome) quantum measurements, the fat-shattering dimension is $\min\left\{O\left(d/\epsilon^2\right),d^2\right\}$. We also showed that the fat-shattering dimension for its dual problem---learning quantum states---is  $\min\left\{O\left(\log d/\epsilon^2\right),d^2-1\right\}$. Our proof is entirely based on tools from classical learning theory, and provides an alternative proof for Aaronson's result \cite{Aar07}.
 We also derived other important complexity measures for these two tasks, and the results are summarized in Table \ref{table:QML}. Our results demonstrated that learning an unknown measurement is a more daunting task than learning an unknown quantum state.  The intuition is that,  since the effect space is much larger than the state space, it is reasonable that the fat-shattering dimension of the effect space is larger, too.

Finally, by exploiting general Bloch-sphere representation, we show that our learning problems are equivalent to a \emph{neural network} so that classical ML algorithms can be applied to learn the unknown quantum measurement or state. Our work could provide a new viewpoint to the study of quantum state and measurement tomography. We  also discuss connections between the quantum  learning problems and other fields in QIP such as existence of QRA Codes and quantum state discrimination.
We hope that the development of our results would stimulate more theoretical studies in quantum statistical learning, and more applications in quantum information processing and  related areas can be discovered.

\appendix
\section{Notation Table} \label{notation}

\begin{table}[h]\center
	\caption{Summary of Notation} \label{table:notation}
	\begin{tabular}{ll}
		\toprule
		\textbf{Notation} & \textbf{Mathematical Meaning}\\
		\midrule
		$\mathcal{H}$ & the (separable) Hilbert space\\
		$d$ & the dimension of the linear space\\
		$\mathbb{R},\,\mathbb{N}$ & the set of real numbers and positive integers\\
		$\mathbb{C}^d$ & the linear space of $d$-dimensional complex vectors\\
		$\mathbb{M}_d$ & the set of all self-adjoint operators on $\mathbb{C}^d$\\
		$\Tr$ & the trace function on $\mathbb{M}_d$\\
		$A^\dagger$ & the conjugate transpose of $A$\\
		$\langle A,B\rangle$ & $=\Tr(B^\dagger A)$, the Hilbert-Schmidt inner product  on $\mathbb{M}_d$; \\	
		& also stands for conventional inner product on $\mathbb{C}^d$\\
		$\mathcal{B(H)}$ & the set of bounded operators on $\mathcal{H}$\\
		$\mathcal{T(H)}$ & the set of trace class operators (i.e.~finite trace) on $\mathcal{H}$\\
		$\mathcal{O}$ & the zero operator on $\mathcal{H}$.\\
		$\mathcal{I}$ & the identity operator on $\mathcal{H}$.\\
		$A \succeq B$ & $=A-B\succeq\mathcal{O}$, the standard partial ordering\\		
		$\|M\|_p$ & the Schatten $p$-norm on $\mathbb{M}_d$, which reduces to the $\ell_p$ norms on $\mathbb{C}^d$.\\
		$S_p^d$ & =$\{M\in\mathbb{M}_d:\|M\|_p\leq 1\}$, the unit ball of Schatten $p$-class\\
		$|\varphi\rangle$ & the unit vector on $\mathcal{H}$\\
		$\rho,\,\sigma$ & the quantum state on $\mathcal{H}$, i.e.~$\rho=\rho^\dagger\in\mathcal{T(H)}$, with $\Tr(\rho)=1$\\
		$E,\,\Pi$ & the POVM element on $\mathcal{H}$, i.e.~$\mathcal{O}$\\
		$\mathcal{Q(H)}$ & state space, the set of all states on $\mathcal{H}$\\
		$\mathcal{E(H)}$ & effect space, the set of all POVM elements on $\mathcal{H}$\\
		$\mathcal{X}$ & the input space, or called the instances domain (the set)\\
		$\mathcal{Y}$ & the output space, or called the labels domain (the set)\\
		$\mathcal{Z}$ & $=\mathcal{X\times Y}$\\
		$\mathcal{F}$ & the hypothesis set of functions $f:\mathcal{X}\rightarrow \mathcal{Y}$\\
		$\mu$ & a distribution on $\mathcal{Z}$\\
		$\mathsf{Z}_n$ & a training data set of $n$ elements independently  according to $\mu$\\
		$\ell_f:\mathcal{Z}\rightarrow (0,\infty)$ & loss function\\
		$\Pr,\,\mathbb{E}$ & probability and expectation of a random variable\\
		$L(f)$ & $=\mathbb{E}_\mu[\ell_f(X,Y)]$, the ensemble error\\
		$\widehat{L}_n(f)$ & $=1/n\sum_{i=1}^n\ell_f(X_i,Y_i)$, the empirical error over the training data set $\mathsf{Z}_n$\\
		VCdim$(\mathcal{F})$ & Vapnik-Chervonenkis dimension of the function class $\mathcal{F}$\\
		Pdim($\mathcal{F})$ & pseudo dimension of the function class $\mathcal{F}$\\
		fat$_\mathcal{F}(\epsilon)$ & fat-shattering dimension of the function class $\mathcal{F}$ with $\epsilon>0$\\
		$\underline{\text{fat}}_\mathcal{F}(\epsilon)$ & level fat-shattering dimension of the function class $\mathcal{F}$ with $\epsilon>0$\\
		$\mathcal{N}(\epsilon,\mathcal{F},\tau)$ & covering number of $\mathcal{F}$ with metric $\tau$ and $\epsilon>0$\\
		$\log\mathcal{N}(\epsilon,\mathcal{F},\tau)$ & entropy number\\
		$\mathcal{R}_n(\mathcal{F})$ & Rademacher complexity of the function class $\mathcal{F}$ on $\mathsf{Z}_n$\\
		$\gamma_i$ & uniformly $\{+1,-1\}$-valued random variables or called Rademacher variables\\
		$O$ & the big O notation; $f=O(g)$ means $f(x)\leq cg(x)$\\
		& for some positive $c,\,x_0$ and all $x\geq x_0$\\
		$A\lesssim B$ & $=A\leq cB$\\
		& for some constant $c$\\
		$A\simeq B$ & both $A\lesssim B$ and $A\gtrsim B$\\
		\bottomrule
	\end{tabular}
\end{table}

\section{Sample Complexity in Terms of Complexity Measure} \label{SC}

In Section \ref{MC}, we introduce several complexity measures.
In this section, we list some well-known deviation formula to express the generalisation error and sample complexity in terms of those complexity measures.

It has been established that any set of Boolean functions is a uGC class (i.e.~PAC learnable) if and only if it has a finite VC dimension \cite{Vap82, BEH+89}. Additionally, the finite VC dimension provides an upper bound for the sample complexity of the Boolean function class.
\begin{theo}[Vapnik \emph{et al.} \cite{Vap82, BEH+89, BM02}] \label{sc_VC}
	Let $C$ be an absolute constant and $\mathcal{F}$ be a class of Boolean functions which has a finite VC dimension $d$. Then, for every $0<\epsilon,\delta<1$,
	
	\[
	\sup_\mu \Pr \left\{ \sup_{f\in\mathcal{F}} \left|L(f)-\widehat{L}_n(f)\right|\geq \epsilon \right\} \leq \delta,
	\]
	provided that $n\geq \frac{C}{\epsilon^2}\left(d\log(2/\epsilon)+\log(2/\delta)\right)$.
	
	Therefore, the sample complexity is bounded by
	\begin{eqnarray} \label{eq:sc_VC}
	m_{\mathcal{F}}(\epsilon,\delta)\leq\frac{C}{\epsilon^2}\max\left\{d\log\frac1{\epsilon},\log\frac1{\delta}\right\}.
	\end{eqnarray}
\end{theo}

Following the same reasoning as in Theorem \ref{sc_VC}, the analogous results can be drawn: the hypothesis set $\mathcal{F}$ is a uGC class if and only if it has a finite fat-shattering dimension for every $\epsilon>0$ \cite{BLW96, ABC+97, MV03}.
We have the following theorem:

\begin{theo}[Bartlett \emph{et al.} \cite{BLW96, ABC+97, MV03}] \label{sc_fat}
	There is an absolute constant $C$ such that for every $\mathcal{F}$ consisting of bounded functions and every $0<\epsilon,\delta<1$,
	\[
	\sup_\mu \Pr \left\{ \sup_{f\in\mathcal{F}} \left|L(f)-\widehat{L}_n(f)\right|\geq \epsilon \right\} \leq \delta,
	\]
	provided that $n\geq \frac{C}{\epsilon^2}\left(\text{fat}_{\mathcal{F}}(\epsilon/8)\cdot\log(2/\epsilon)+\log(8/\delta)\right)$.
	
	Therefore, the sample complexity is bounded by
	\begin{eqnarray} \label{eq:sc_fat}
	m_{\mathcal{F}}(\epsilon,\delta)\leq\frac{C}{\epsilon^2}\max\left\{\text{fat}_{\mathcal{F}}(\epsilon)\cdot\log\frac1{\epsilon},\log\frac1{\delta}\right\}.
	\end{eqnarray}
\end{theo}

The entropy number is distribution-independent and is closely related to the learnability of the function class. Dudley et al. \cite{DGZ91} showed that a class $\mathcal{F}$ consisting of bounded functions is a uGC class if and only if that there is some $1\leq p\leq \infty$ such that for every $\epsilon>0$,
\[
\lim_{n\rightarrow\infty} \frac{\log \mathcal{N}_p(\epsilon,\mathcal{F},n)}n =0.
\]
In addition, we have the following theorem:
\begin{theo}[Polland \cite{Pol84}] \label{sc_cn}
	Let $\mathcal{F}$ be a set of bounded functions.
	\begin{itemize}
		\item[(i)] For every $0<\epsilon<1$, any $n\geq 8/\epsilon^2$, and any probability measure $\mu$,
		\[
		\Pr \left\{ \sup_{f\in\mathcal{F}} \left|L(f)-\widehat{L}_n(f)\right|\geq \epsilon \right\} \leq 8\mathcal{N}_1(\epsilon/8,\mathcal{F},n)\exp(-\frac{n\epsilon^2}{128}).
		\]
		\item[(ii)] For every $0<\epsilon,\delta<1$,
		\[
		\sup_\mu\Pr \left\{ \sup_{f\in\mathcal{F}} \left|L(f)-\widehat{L}_n(f)\right|\geq \epsilon \right\}
		\leq \delta,
		\]
		provided that $n\geq \frac{C}{\epsilon^2}\left( \log \mathcal{N}_1(\epsilon,\mathcal{F})+\log(2/\delta)\right)$.
		
		Therefore, the sample complexity is bounded by
		\begin{eqnarray} \label{eq:sc_cn}
		m_{\mathcal{F}}(\epsilon,\delta)\leq\frac{C}{\epsilon^2}\max\left\{\mathcal{N}_1(\epsilon,\mathcal{F}),\log\frac1{\delta}\right\}.
		\end{eqnarray}
	\end{itemize}
\end{theo}

\begin{theo}[Bartlett and Mendelson \cite{BM02}]\label{sc_ra}
	For any $0<\delta<1$, with probability at least $1-\delta$ and for all $f\in\mathcal{F}$ we have,
	\[
	\Pr \left\{ \sup_{f\in\mathcal{F}} \left|L(f)-\widehat{L}_n(f)\right|\geq \epsilon \right\} \leq \delta,
	\]
	provided that $n\geq \frac{C}{\epsilon^2}\max\left\{\mathcal{R}_n(\mathcal{F}),\log(1/\delta)\right\}$.
	
	Therefore, the sample complexity is bounded by
	\begin{eqnarray} \label{eq:sc_ra}
	m_{\mathcal{F}}(\epsilon,\delta)\leq \frac{C}{\epsilon^2}\max\left\{\mathcal{R}_n(\mathcal{F}),\log\frac1{\delta}\right\}
	\end{eqnarray}
\end{theo}

\section{Learning Framework in Bloch-sphere Representation} \label{BR}
When illustrating the state space on a finite dimensional Hilbert space $\mathbb{C}^d$, it is convenient to adopt a geometric parameterisation method called \emph{Bloch-sphere representation} \cite{BZ08, Kim03, KK05}. Here, we provide another point of view on our quantum learning framework. The key idea is to represent the quantum objects in a Euclidean space, wherein classical techniques of traditional ML can be applied.
Although the Bloch-sphere representation method may not be as direct as the machinery we used in Sections \ref{LQE} and \ref{LQS}, it does gain more insights into our quantum ML problems.

Based on the orthogonal basis $\{\mathcal{I},\Lambda_1, \ldots, \Lambda_{d^2-1}\}$ of $SU(d)$, any state $\rho_\mathbf{r}$ on $\mathbb{C}^d$ can be represented in a \emph{Bloch vector} $\mathbf{r}$ through:
\begin{eqnarray} \label{eq:bloch_vec}
\rho_\mathbf{r}=\frac1d \left( \mathcal{I}+c_d \sum_{i=1}^{d^2-1} r_i \Lambda_i \right) = \frac1d (\mathcal{I} + c_d \mathbf{r}\cdot\mathbf{\Lambda}),
\end{eqnarray}
where $c_d:=\sqrt{\frac{d(d-1)}{2}}$ and the dot product corresponds to the conventional Euclidean inner product, and
\[
r_i=\sqrt{\frac{d}{2(d-1)}} \Tr \left(\rho_\mathbf{r}\Lambda_i\right)\in\mathbb{R},\, i=1,\ldots,d^2-1.
\]
Define the Bloch vector space as the set of Bloch vectors, which are representative of the valid states on $\mathbb{C}^d$ as
\[
\Omega_{d}:=\{\mathbf{r}\in\mathbb{R}^{d^2-1}:\, \mathbf{r}=\sqrt{\frac{d}{2(d-1)}} \Tr \left(\rho_\mathbf{r}\cdot\mathbf{\Lambda}\right)\}.
\]

Now we calculate the linear functional of
$E_\mathbf{n}\in\mathcal{E}_1$ acting on the state $\rho_\mathbf{r}$
(where $\mathcal{E}_k$ denotes the convex hull of rank-$k$
projection operators):
\begin{align*}
\Tr(P_\mathbf{n} \rho_\mathbf{r})&= \Tr \left(\frac{1}{d^2} (\mathcal{I}+c_d \mathbf{r}\cdot\mathbf{\Lambda})(\mathcal{I}+c_d \mathbf{n}\cdot\mathbf{\Lambda})\right) \\
&=\Tr\left(\frac{1}{d^2} [\mathcal{I}+c_d (\mathbf{r}\cdot\mathbf{\Lambda}+\mathbf{n}\cdot\mathbf{\Lambda})+c_d^2(\mathbf{r}\cdot\mathbf{\Lambda})(\mathbf{n}\cdot\mathbf{\Lambda})]\right)\\
&=\frac1d+\frac{c_d^2}{d^2}\Tr\left( (\mathbf{r}\cdot\mathbf{\Lambda})(\mathbf{n}\cdot\mathbf{\Lambda}) \right)\\
&= \frac1d \left( 1+(d-1)\mathbf{r}\cdot\mathbf{n}\right).
\end{align*}
Consequently, we have the affine functionals with elements in the convex hull of rank-1 projection operators, i.e.~
\[
\mathcal{F}_1=\{\rho_\mathbf{r}\mapsto \frac1d(1+(d-1)\mathbf{r}\cdot\mathbf{n}):\, \mathbf{n}\in\Omega_d\}.
\]

In order to characterise the quantum effects associate with higher dimensional projection operators, it is useful to consider the algebraic properties of the projection operators. The set of projection operators on $\mathbb{C}^d$ is not a vector space but corresponds to an orthocomplemented lattice. Therefore, the sum of two projections, say $P$ and $Q$, is a projection only when they are orthogonal, i.e.~$PQ=QP=\mathcal{O}$. Based on this fact, now let $\{P_{\mathbf{n}_1}, \ldots, P_{\mathbf{n}_d}\}$ be arbitrary mutually orthogonal rank-one projections on $\mathbb{C}^d$. To each of them, we associate a unit Bloch vector $\mathbf{n}_i$ such that $P_{\mathbf{n}_i}=\frac1d(\mathcal{I}+ c_d \mathbf{n}_i\cdot \mathbf{\Lambda}),\, i=1,\ldots,d$.
It can be verified by Eq.~\eqref{eq:bloch_vec} that the Bloch vectors $\{\mathbf{n}_1,\ldots,\mathbf{n}_d\}$ form a $(d-1)$-dimensional (regular) simplex since
the angle between any two Bloch vectors is $\theta(\mathbf{n}_i,\mathbf{n}_j)=\cos^{-1}(-\frac1{d-1})$.
With a slight abuse of notation, denote a rank-$k$ projection $P_{\mathbf{n}_{(k)}}$ as the summation of arbitrary $k$ different projections from the set $\{P_{\mathbf{n}_1}, \ldots, P_{\mathbf{n}_d}\}$. More formally, we denote an index set $I_k\subseteq \{1,\ldots,d\}$ with cardinality $k$, and $P_{\mathbf{n}_{(k)}}=\sum_{i\in I_k} P_{\mathbf{n}_i}$, where we adopt the convention that the empty sum is zero. Hence, when a rank-$k$ projection $P_{\mathbf{n}_{(k)}}\in\mathcal{F}_k$ acts on the state $\rho_\mathbf{r}$, we have:
\begin{align}\label{eq:rank_k}
\Tr(P_{\mathbf{n}_{(k)}}\rho_\mathbf{r})&= \sum_{i\in I_k} \frac1d(1+(d-1)\mathbf{r}\cdot\mathbf{n}_i)\notag \\
&= k\cdot\frac{1}{d}(1+(d-1)\mathbf{r}\cdot\mathbf{n}_{(k)}),
\end{align}
where 
\begin{eqnarray} \label{eq:nk}
\mathbf{n}_{(k)}:=\frac1k \sum_{i\in I_k} \mathbf{n}_i
\end{eqnarray}
 is the centroid of the $(k-1)$-face of the simplex $\Delta_{d-1}$ subtended by the vectors $\{\mathbf{n}_i\}_{i\in I_k}$.
The $\ell_2$-norm of $\mathbf{n}_{(k)}$ can be calculated as the Euclidean distance from the center of the simplex $\Delta_{d-1}$ to the centroid of $(k-1)$-face; that is
\begin{eqnarray} \label{nk}
\|\mathbf{n}_{(k)}\|_2:=r_{d,k}=\sqrt{\frac{d-k}{k(d-1)}}<1,\, k\in\{1,2,\ldots,d\}.
\end{eqnarray}
Intuitively, we can interpret the value $\Tr(P_{\mathbf{n}_{(k)}}\rho_\mathbf{r})$ as an operator $P_{\mathbf{n}_{(k)}}$ acting on the state $\rho_\mathbf{r}$, and then scaled by $k$.

Since every quantum effect can be composed into the extremal effects (i.e.~projection operators) of the effect space \cite{Dav76}.
We can represent Tr$(E_\mathbf{n}\rho_\mathbf{r})$ for all $E_\mathbf{n}\in\mathcal{E}(\mathbb{C}^d)$ as:
\begin{eqnarray}\label{eq:functional_effect}
\sum_{k=0}^d \frac{w_k\cdot
k}{d}\left(1+(d-1)\mathbf{r}\cdot\mathbf{n}_{(k)}\right)=\frac{1}{d}\left(n_0+(d-1)\mathbf{r}\cdot\mathbf{n}\right),
\end{eqnarray}
where $\sum_{i=0}^d w_k=1$, $0\leq n_0 \leq d$ and
$\|\mathbf{n}\|_2\leq \max_{k\in\{0,1,\ldots,d\}}
\sqrt{\frac{k(d-k)}{d-1}}$.

By utilising the bijection relationship of quantum state $\rho_\mathbf{r}$ and its corresponding Bloch vectors $\mathbf{r}$, we can associate the input space as the Bloch vector space, i.e.~$\mathcal{X}=\Omega_d$.
Denote the function class $\mathcal{F}_k$ as the linear functionals of $\mathcal{E}_k$ acting on $\rho_\mathbf{r}$.
According to Eq.~\eqref{eq:rank_k}, we have:
\begin{eqnarray} \label{eq:rank_k_bloch}
\mathcal{F}_k=\text{conv}\left(\{\mathbf{r}\mapsto\frac{k}{d}\left(1+(d-1)\mathbf{r}\cdot\mathbf{n}_{(k)}\right)\}\right).
\end{eqnarray}
For the rank-$0$ projection operator, the class consists of only one element, i.e.~$\mathcal{F}_{0}=\{\mathcal{O}\}$.
We can see from the above equation that the affine coefficient is fixed such that $\mathcal{F}_k$ consists of linear functionals.
For the class of all quantum effects $\mathcal{F}=\mathcal{E}(\mathbb{C}^d)$, by Eq.~\eqref{eq:functional_effect} we have a similar result:
\[ 
\mathcal{F}=\{\mathbf{r}\mapsto\frac{1}{d}\left(
n_0+(d-1)\mathbf{r}\cdot\mathbf{n}\right)\,:\:
\mathbf{n}\in\mathbb{R}^{d^2-1}\},\, \mathbf{r}\in\Omega_d,
\]
where $n_0$ can be upper bounded by $d$ and $\|\mathbf{n}\|_2$ can
be bounded by $k\cdot r_{d,k}=\sqrt{\frac{k(d-k)}{d-1}}$. Clearly,
$\mathcal{F}=\mathcal{E(H)}$ is the function class consisting of the
affine functionals. However, we can easily convert this formulation
into a linear form by letting $\tilde{\mathbf{r}}=[1, \mathbf{r}]$,
and $\tilde{\mathbf{n}}=[n_0, \mathbf{n}]$.
The intuition behind this is that when characterising the learnability of quantum measurements, all we need is to bound the complexity measures of the class of linear functionals. 

\section{Neural Networks} \label{nn}
Here we briefly introduce the theory of \emph{Neural Networks}. Readers may refer to Ref. \cite{AB99} for more details. The basic computing unit in a neural network is the (simple) \emph{perceptron} (see Fig.~\ref{fig:single}), which computes a function from $\mathbb{R}^d$ to $\mathbb{R}$:
\[
f(\mathbf{r})=\sigma(\mathbf{v}\cdot\mathbf{r}+v_0),
\]
for input vector $\mathbf{r}\in\mathbb{R}^d$, where $\mathbf{v}=(v_1,\ldots,v_d)\in\mathbb{R}^d$ and $v_0\in\mathbb{R}$ are adjustable parameters, or \emph{weights} (the particular weight $v_0$ being known as the \emph{threshold}). The function $\sigma:\mathbb{R}\rightarrow\mathbb{R}$ is called the \emph{activation function}. In the scenario of binary classification, the activation function may be chosen as the sign function; in the case of real-value outputs, $\sigma(\cdot)$ may satisfy some Lipschitz conditions. Note that the decision boundary of the binary perceptrons is the affine subspace of $\mathbb{R}^d$ defined by the equation $\mathbf{v}\cdot\mathbf{r}+v_0=0$.

When using a simple perceptron for a binary classification problem, the \emph{perceptron learning algorithm} (PCA) finds adequate parameters $\mathbf{v}$ and $v_0$ to well fit the training data set. The algorithm starts from an arbitrary initial parameter and updates the parameter when there are misclassified data. For example, if now the function computes $(\mathbf{r},y)$ (with $\mathbf{r}\in\mathbb{R}^d$ and $y\in \{0,1\}$), the algorithm adds $\eta(y-f(\mathbf{r}))[\mathbf{r}, -1]$ element-wise to $[\mathbf{v}, v_0]$, where $\eta$ is a fixed step constant. PCA iterates until a termination criterion is reached.

The second example is the \emph{two-layer networks} (also called \emph{single-hidden layer nets}) (see Fig.~\ref{fig:two}). The network can compute a function of the form
\[
f(\mathbf{r})=\sum_{i=1}^k w_k \sigma(\mathbf{v}_i\cdot\mathbf{r}+v_{0i})+w_0,
\]
where $w_i\in\mathbb{R}, i=0,\ldots,k,$ are the output weights, $[\mathbf{v}_i, v_{0i}]$ are the input weights. The positive integer $k$ is the number of hidden units. One can use a `gradient descent' procedure to adjust the parameters to minimize the squared errors over the training data.

\begin{figure}[htb]
	\caption{Consider a qubit system. A measurement in $\mathcal{F}_1$ can be characterised by a simple perceptron with $3$-dimensional input data and the activation function $\sigma$. The `1' node is a bias node and $v_0$ is the corresponding bias weight. The input vector is the Bloch vector $\mathbf{r}\in\Omega_2$. The output variable $y=f(\mathbf{r})$ is computed by the simple perceptron. Hence the problem of learning an unknown measurement $\Pi\in\mathcal{F}_1$ is to infer the simple perceptron, i.e.~the values of $\{v_i\}_{i=1}^4$.}\label{fig:single}
	\begin{neuralnetwork}[height=4]
		\newcommand{\nodetextclear}[2]{}
		\newcommand{\nodetextx}[2]{\ifnum0=#2 $1$\else $r_#2$ \fi}
		\newcommand{\nodetexty}[2]{$y$}
		\newcommand{\nodetextSum}[2]{$\Sigma$}
		\newcommand{\nodetextAct}[2]{$\sigma$}
		\newcommand{\linkA}[4]{$v_#2$}
		\inputlayer[count=3, bias=true, title=Input Vector, text=\nodetextx]
		\hiddenlayer[count=1, bias=false, nodeclass={hidden neuron},title=Summation, text=\nodetextSum]
		\link[from layer=0, to layer=1, from node=0, to node=1, label=\linkA]
		\link[from layer=0, to layer=1, from node=1, to node=1, label=\linkA]
		\link[from layer=0, to layer=1, from node=2, to node=1, label=\linkA]
		\link[from layer=0, to layer=1, from node=3, to node=1, label=\linkA]
		\hiddenlayer[count=1, bias=false, title=Activation, text=\nodetextAct]
		\linklayers
		\outputlayer[count=1, title=Output Variable, text=\nodetexty] \linklayers
	\end{neuralnetwork}
\end{figure}

\begin{figure}[htb]
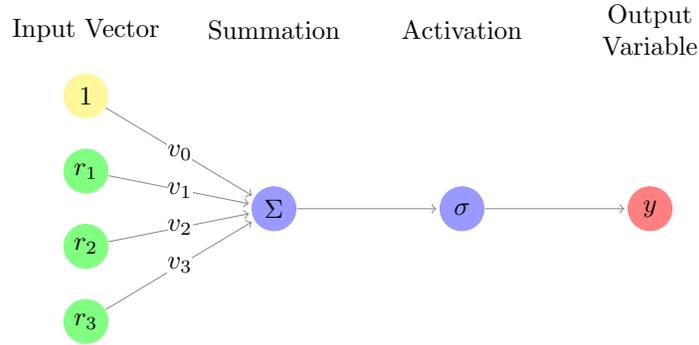

	\caption{Single-hidden layer net computes $3$-dimensional input data with activation function $\sigma$ and three hidden units, which correspond to $\mathcal{F}_i$ for $i=0,1,2$. The value $v_{0k}$ corresponds to the bias weight of the $k$-th hidden unit. The single-hidden net represents a quantum measurement in $\mathcal{E}(\mathbb{C}^2)$.} \label{fig:two}
	\begin{neuralnetwork}[height=4]
		\newcommand{\nodetextclear}[2]{}
		\newcommand{\nodetextx}[2]{\ifnum0=#2 $1$\else $r_#2$ \fi}
		\newcommand{\nodetexty}[2]{$y$}
		\newcommand{\nodetextSum}[2]{$\Sigma$}
		\newcommand{\nodetextAct}[2]{\ifnum0=#2 $1$\else $\sigma$ \fi}
		\newcommand{\linkA}[4]{$v_{#2#4}$}
		\newcommand{\linkB}[4]{$w_{#2}$}
		\inputlayer[count=3, bias=true, title=Input Vector, text=\nodetextx]
		\hiddenlayer[count=4, bias=false, exclude={1}, title=Summation, text=\nodetextSum]
		\link[from layer=0, to layer=1, from node=0, to node=2, label=\linkA]
		\link[from layer=0, to layer=1, from node=0, to node=3, label=\linkA]
		\link[from layer=0, to layer=1, from node=0, to node=4, label=\linkA]
		\link[from layer=0, to layer=1, from node=1, to node=2]
		\link[from layer=0, to layer=1, from node=1, to node=3]
		\link[from layer=0, to layer=1, from node=1, to node=4]
		\link[from layer=0, to layer=1, from node=2, to node=2]
		\link[from layer=0, to layer=1, from node=2, to node=3]					
		\link[from layer=0, to layer=1, from node=2, to node=4]
		\link[from layer=0, to layer=1, from node=3, to node=2]
		\link[from layer=0, to layer=1, from node=3, to node=3, label=\linkA]					
		\link[from layer=0, to layer=1, from node=3, to node=4, label=\linkA]
		\hiddenlayer[count=3, bias=true, title=Hidden Layer, text=\nodetextAct]
		\link[from layer=1, to layer=2, from node=2, to node=1]
		\link[from layer=1, to layer=2, from node=3, to node=2]
		\link[from layer=1, to layer=2, from node=4, to node=3]
		\hiddenlayer[count=1, bias=false, title=Summation, text=\nodetextSum]
		\link[from layer=2, to layer=3, from node=0, to node=1, label=\linkB]					
		\link[from layer=2, to layer=3, from node=1, to node=1, label=\linkB]					
		\link[from layer=2, to layer=3, from node=2, to node=1, label=\linkB]
		\link[from layer=2, to layer=3, from node=3, to node=1, label=\linkB]			
		\outputlayer[count=1, bias=false, title=Output Variable, text=\nodetexty]
		\linklayers					
	\end{neuralnetwork}	
\end{figure}

\printbibliography

\end{document}